\documentclass[11pt,letterpaper]{article}
\usepackage[utf8]{inputenc}
\usepackage[margin=1in]{geometry}
\usepackage{amsmath}
\usepackage{palatino}
\usepackage{macros}
\usepackage{latexsym} 
\usepackage{enumitem}
\allowdisplaybreaks

\bibliography{refs}

\DeclareMathOperator{\dimfunc}{dim}
\DeclareMathOperator{\codimfunc}{codim}
\DeclareMathOperator{\offdiag}{off-diag}

\newcommand{\declareperson}[1]{\expandafter\newcommand\csname#1\endcsname[1]{\textcolor{orange}{#1: ##1}}}

\declareperson{Kuikui}
\declareperson{Shayan}
\declareperson{Dorna}

\newcommand{\edgecodim}{k}
\newcommand{\dimension}{d}
\newcommand{\boundMatrix}{M}
\newcommand{\dbound}{F}

\newcommand{\degree}{\Delta}
\newcommand{\maxDeg}{\Delta}
\newcommand{\extraColors}{\beta}
\newcommand{\linear}{\bar{A}}
\newcommand{\sqr}{S}

\def\eps{\epsilon}
\def\bone{{\bf 1}}
\DeclareMathOperator{\Ent}{Ent}
\DeclareMathOperator{\mix}{mix}
\DeclareMathOperator{\TV}{TV}

\usepackage{authblk}

\title{A Matrix Trickle-Down Theorem on Simplicial Complexes\\ and Applications to Sampling Colorings}

\author{Dorna Abdolazimi\thanks{\href{mailto:dornaa@cs.washington.edu}{dornaa@cs.washington.edu}. Research supported by CCF-1907845 and AFOSR grant FA9550-20-1-0212.}}
\author{Kuikui Liu\thanks{\href{mailto:liukui17@cs.washington.edu}{liukui17@cs.washington.edu}. Research supported in part by NSF grants CCF-1552097, CCF-1907845.}}
\author{Shayan Oveis Gharan\thanks{\href{mailto:shayan@cs.washington.edu}{shayan@cs.washington.edu}. Research supported by NSF grants  CCF-1552097, CCF-1907845, AFOSR grant FA9550-20-1-0212 and a Sloan fellowship.}} 
\affil{University of Washington}

\begin{document}

\maketitle

\begin{abstract}
We show that the natural Glauber dynamics mixes rapidly and generates a random proper {\em edge}-coloring of a graph with maximum degree $\Delta$ whenever the number of colors is at least $q\geq \left(\frac{10}{3} + \eps\right)\Delta$, where $\eps>0$ is arbitrary and the maximum degree satisfies $\Delta \geq C$ for a constant $C = C(\epsilon)$ depending only on $\epsilon$. For edge-colorings, this improves upon prior work \cite{Vig99, CDMPP19}  which show rapid mixing when $q\geq \left(\frac{11}{3}-\eps_{0}\right) \Delta$, where $\epsilon_{0} \approx 10^{-5}$ is a small fixed constant. 
At the heart of our proof, we establish a matrix trickle-down theorem, generalizing Oppenheim's influential result, as a new technique to prove that a high dimensional simplicial complex is a local spectral expander.
\end{abstract}

\section{Introduction}
Given an (undirected) graph $G=(V,E)$ with $n=|V|$ vertices and with maximum degree $\Delta\geq 1$ can we generate a uniformly random proper coloring of vertices of $G$ using $q$ colors?
For $q\leq \Delta$, there is no efficient algorithm (in the sense of an $\mathsf{FPRAS}$) to approximately count proper $q$-colorings (at least when $q$ is even) unless $\mathsf{NP} = \mathsf{RP}$, even for $\Delta$-regular graphs which are triangle-free \cite{GSV15}. This fundamental question in the field of counting and sampling has puzzled researchers for decades. One can study a natural Markov Chain (MC), known as the Glauber dynamics, to generate a random proper coloring: Given a proper vertex coloring of $G$, choose a uniformly random vertex $v\in G$ and re-color $v$, namely choose a uniformly random color which is not present in any of the neighbors of $v$.

It is not hard to see that for $q\geq \Delta+2$ this chain is irreducible and has unique stationary distribution which is uniform over all proper colorings of $G$. It is conjectured that the Glauber dynamics mixes in time $O(n\log n)$ for $q$ as low as $\Delta+2$ but despite significant attempts we are still very far from proving this conjecture.

To this date, the best known result for general graphs is due Chen, Delcourt, Moitra, Perarnau and Postle \cite{CDMPP19} who show that the Glauber dynamics mixes in polynomial time for $q\geq (11/6-\eps)\Delta$ for some universal constant $\eps>0$; this slightly improves on the classical works of Jerrum and Vigoda \cite{Jer95, Vig99} which bounds the mixing time by a polynomial in $n$ for $q\geq (11/6)\Delta$.

Most of the recent analyses of the Glauber dynamics  are focused on ``locally sparse'' graphs \cite{HV03,Mol04,HV05,FV06,FV07,HVV07,DFHV13, CGSV21,FGYZ21} where it was typically shown how to break the $11/6$ barrier bound when the underlying graph  has a large girth.  We note that these assumptions are typically very strong as it can be seen that triangle graph free graphs can be colored with as little as $O(\Delta/\log\Delta)$ many colors \cite{John96}.

The results on locally sparse graphs typically exploit (strong) correlation decay properties: Roughly speaking, they imply that if we color a vertex $v$ with a color $c$, then the marginal probability of  coloring a ``far away'' vertex $u$ with a color $c'$ does not change, or changes very mildly. Although it is conjectured that vertex coloring exhibits correlation decay, more formally known as ``strong spatial mixing'' property, for $q\geq \Delta+O(1)$, to this date, we are lacking techniques to establish such a statement (see e.g., \cite{GMP05,Y14,GKM15,EGHSV19}).

In this paper, we study random proper edge coloring of graphs, which equivalently can be seen as a random proper vertex coloring of line graphs. Unlike most recent trends which focus on sparse graphs with large girths, line graphs are very dense locally as they contain induced cliques of size $\Omega(\Delta)$. To the best of our knowledge, the only previous result on edge coloring which goes significantly beyond the $11/6$ barrier is recent work of Delcourt, Heinrich and Perarnau \cite{DHP20} which shows that the Glauber dynamics mixes rapidly when the underlying graph is a tree and $q\geq  \Delta+1$. Note that for a graph with maximum degree $\Delta$, the maximum degree of the line graph could be as large as $2\Delta$. Therefore, with the $11/6$ barrier, one would need  $q \geq 11/3 \Delta$ to guarantee polynomial mixing for all edge coloring instances. In our main theorem we prove that this barrier can be broken for edge coloring of any graph with maximum degree $\Delta$.

\begin{theorem}[Main]\label{thm:glauberedge}
Let $G=(V,E)$ be a graph of maximum degree $\Delta$. For any $0<\epsilon\leq \frac{1}{10}$  such that
$\frac{\ln^{2}\Delta}{\Delta} \leq \frac{\epsilon^3}{15}$, and any collection of color lists $L = \{L(e)\}_{e \in E}$ satisfying $|L(e)| \geq \Delta(e)+(4/3+4\epsilon)\Delta$ where $\Delta(e)$ is the number of neighbors of $e$ in the line graph of $G$, the spectral gap of  the  Glauber dyanmics for sampling proper $L$-edge-list-colorings on $G$ is $\Omega(n^{-O(1/\eps)})$, so the mixing time is $O(n^{O({1/\eps})})$.
Furthermore, if $ \Delta \leq O(1)$, the modified and standard log-Sobolev constants are $\Omega_{\epsilon,\Delta}(1/n)$. 
\end{theorem}

We remark that our general mixing time bound has no dependence on $\Delta$ or $q$. So, the algorithm runs in polynomial time even for graphs of unbounded degree.

In a second contribution we show that for any list vertex coloring instance where $G$ is a tree with max degree $\Delta$, and the size of the list of every vertex $v$ is at least $\Delta(v)+\eps \Delta$ for $\eps=\Omega(\frac{\ln\Delta}{\sqrt{\Delta}})$ the Glauber dyanics mixes rapidly and generates a uniformly random vertex coloring of $G$.
\begin{theorem}\label{thm:glaubertrees}
Let $G=(V,E)$ be a tree of maximum degree $\Delta$. For any  $0<\epsilon \leq  1$ such that  $\frac{\ln^{2}\Delta}{\Delta} \leq \frac{\epsilon^{2}}{100}$ and  any collection of color lists $L = \{L(v)\}_{v \in V}$ satisfying $|L(v)| \geq \Delta(v) + \epsilon \Delta$, the spectral gap  of the Glauber dynamics for sampling proper $L$-vertex-list-colorings on $G$ is $\Omega(n^{-O(1/\eps)})$, therefore the mixing time is $O(n^{O(1/\eps)})$. Furthermore, if $ \Delta \leq  O(1)$, the modified and standard log-Sobolev constants are $\Omega_{\epsilon,\Delta}(1/n)$. 
\end{theorem}

The above theorem, although it is not  as strong as \cite{MSW04}, shows that Glauber dynamics mixes rapidly even when we have a list coloring problem on a tree and furthermore it gives a possible avenue to exploit our techniques to prove that Glauber dynamics mixes rapidly on any graph when $q\geq (1+\eps)\Delta$. We expect that upon further investigation our techniques can be coupled with the extensive literature on random proper colorings of graphs with large girth to break the $11/6-\eps$ barrier.

To establish the above results, we view the Glauber dynamics as a high dimensional walk on a simplicial complex and we prove that this complex is a ``local spectral expander'' (see \cref{def:localspectralexpansion} below). By local-to-global theorems,  local spectral expansion immediately implies a bound on the  spectral gap  of the Glauber dynamics and therefore on its mixing time  (see \cref{subsec:local-to-global} and \cref{subsec:markovchain} below). A simplicial complex $X$  is a downward closed collection of sets. Given a graph $G$ and a set of $q$ colors, we build a simplicial complex as follows: We let each maximal set in $X$ be a set of vertex color pairs that denote a proper coloring of vertices of $G$, then we add all subsets of the maximal sets to $X$ to make it downward closed. By viewing the  Glauber dynamics can be viewed as a walk on the maximal sets in the coloring complex $X$. 
Unlike recent developments \cite{ALO20, CGSV21, FGYZ21, CLV20, CLV21, FGKP21} which exploit the correlation decay property to prove local spectral expansion, our proof is a direct method: We use an inductive argument  inspired by the Oppehheim's trickle-down theorem to bound local eigenvalues. We expect our technique to find more applications in analysis of  Markov chains as well as in other applications of simplicial complexes.

\subsection{Main Technical Contributions}
A simplicial complex $X$ on a finite ground set $U$ is a downwards closed collection of subsets of $U$, i.e. if $\tau \in X$ and $\sigma \subset \tau \subseteq U$, then $\sigma \in X$. The elements of $X$ are called faces, and then maximal faces are called facets.  
We say a face $\tau$ is of dimension $k$ if  $|\tau| = k+1$ and write $\dimfunc (\tau) = k$. We write $X(k) = \{\tau \in X : \dimfunc(\tau) = k\}$. 
We say $X$ is a pure $\dimension$-dimensional complex  if every facet  has dimension $\dimension$.
In this paper we only work with pure simplicial complexes. To keep the notation concise, let $X( \leq  i) := X (-1)\cup\dots\cup X(i)$.
The co-dimension of a face $\tau$ is defined as $\codimfunc (\tau)  \coloneqq \dimension - \dimfunc(\tau)$.  
 For a face $\tau$, define the  link of $\tau$ as the simplicial complex $X_{\tau} = \{\sigma \setminus \tau : \sigma \in X, \sigma \supset \tau\}$. Note that $X_\tau$ is a $(\codimfunc(\tau) -1)$- dimensional complex.
Let $\pi_{X, d}$ be a distribution on the maximal faces of a pure $\dimension$-dimensional simplicial complex $X$; we may drop $X$ from the subscript if it is clear in the context. We will refer to the pair $(X,\pi_{d})$ as a weighted simplicial complex. For each face $\tau$ and each integer $-1 \leq i \leq \codimfunc (\tau)-1$, we write $\pi_{\tau,i}$ for the induced distribution on $X_{\tau}(i)$ given by
\begin{align}\label{eq:pidef}
    \pi_{\tau,i}(\eta) = \frac{1}{\binom{\codimfunc (\tau) }{i+1}} \Pr_{\sigma \sim \pi_{d}}[\sigma \supset \eta \mid \sigma \supset \tau].
\end{align}
One should view this as a marginal distribution conditioned on $\tau$. We will often view $\pi_{\tau,i}$ as a vector in $\R_{\geq0}^{X_{\tau}(i)}$. We remove the subscript $\tau$ when $\tau = \emptyset$. We also omit $i$ from the subscript  when $i = 0$, i.e $\pi_{\tau} \coloneqq \pi_{\tau, 0}$. 

Given a weighted simplicial complex $(X,\pi_{\dimension})$, there is a natural Markov chain known as the down-up walk (or high-order walk) \cite{KM17, DK17, KO18}  whose stationary distribution is $\pi_\dimension$. Starting at a facet $\sigma \in X(\dimension)$, we transition to the next facet $\sigma' \in X(\dimension)$ via the following two-step process:
\begin{enumerate}
    \item Select a uniformly random element $x \in \sigma$ and remove $x$ from $\sigma$.
    \item Select a random $\sigma' \in X(\dimension)$ containing $\sigma \setminus \{x\}$ with probability proportional to $\pi_{\dimension}(\sigma')$.
\end{enumerate}
If $P_{\dimension}^{\vee}$ denotes the transition probability matrix of this Markov chain, then we may write down its entries as follows.
\begin{align*}
    P_{\dimension}^{\vee}(\sigma,\sigma') = \begin{cases}\frac{1}{d+1} \pi_{\sigma\cap\sigma',0}(\sigma'\setminus \sigma) & \text{if } |\sigma\cap \sigma'|=d,\\
    0 &\text{otherwise.}
    \end{cases}
\end{align*}


One of the beautiful insights of \cite{KM17, DK17, KO18} is that spectral properties of $P_{\dimension}^{\vee}$ may be studied through spectral properties of ``local walks'', which are defined using only pairwise conditional marginals and are more tractable to analyze. 
Fix a face $\tau$  of co-dimension $ k \geq 2$; the local walk for $\tau$ is a Markov chain on $X_{\tau}(0)$ with transition probability matrix $P_{X,\tau} \in \R_{\geq0}^{X_{\tau}(0) \times X_{\tau}(0)}$ described by
\begin{align}\label{eq:localwalkdef}
    P_{X,\tau}(x,y) = \frac{\pi_{\tau,1}(x,y)}{2\pi_{\tau}(x)} = \pi_{\tau \cup \{x\}}(y) = \frac{1}{k -1} \Pr_{\sigma \sim \pi_{\dimension}}[y \in \sigma \mid \sigma \supset \tau \cup \{x\}]. 
\end{align}
for distinct $x,y \in X_{\tau}(0)$; we note that $P_{X,\tau}$ has zero diagonal. We drop $X$ in the subscript when it is clear in the context. Note that, by definition, the walk is reversible with stationary distribution $\pi_{\tau}$.  Furthermore, the row corresponding to element $x \in \supp \pi_{\tau}$ is precisely $\pi_{\tau \cup \{x\}}$.
For our calculations, it is easier to work  with matrices $\{P_\tau\}_{X(\leq d-2)}$ if they all live in   $\mathbb{R}^{X(0) \times X(0)}$. 
 Therefore, we modify the definition of $P_\tau$  as follows: for any $\tau \in X(\leq d-2)$, let $P_\tau \in \mathbb{R}^{X(0) \times X(0)}$ be supported on  $X_\tau(0) \times X_\tau (0)$ block and 
equal  to the transition probability  matrix of the local walk for $\tau$ on this block. Similarly, 
we often need to see $\pi_\tau$ as a probability distribution over $X(0)$ supported on $X_\tau(0)$ entries.  Therefore,  $\pi_\tau$ can be seen as a vector in $\mathbb{R}^{X(0) }$. Furthermore, define $\Pi_\tau \in \mathbb{R}^{X(0) \times X(0)}$ as $\Pi_\tau\coloneqq \diag (\pi_\tau) $.  
We proceed by introducing some terminology related to the local walks. 
We say $X$ is totally connected if  $\lambda_2(P_\tau) < 1 $  for any $\tau \in X(\leq d-2)$, i.e. the local walk for $\tau$ is irreducible/connected for any face $\tau$ of co-dimension at least $2$. Furthermore, we  define local spectral expansion as follows. 
\begin{definition}[Local Spectral Expansion \cite{DK17,KO18,DDFH18}   ]\label{def:localspectralexpansion}
We say a weighted simplicial complex $(X,\pi_{\dimension})$ of dimension-$\dimension$ is a $(\gamma_{-1},\dots,\gamma_{\dimension-2})$-local spectral expander if for every $k$ and every $\tau \in X(k)$,  $\lambda_{2}(P_{\tau}) \leq \gamma_{k}$.
\end{definition}
\begin{remark}
We only work with this ``one-sided'' version of local spectral expansion. Many other works \cite{KO18, Opp18, KO18b} consider the stronger two-sided local spectral expansion, which further imposes a lower bound on the smallest eigenvalue of the local walks $P_{\tau}$.
\end{remark}

In almost all applications of high dimensional simplicial complexes, one first needs to to prove that the underlying complex is a local spectral expander and then exploit ``local-to-global'' theorems to prove global properties of the underlying complex $X$. Perhaps, the main generic technique to prove that a given complex $X$ is a local spectral expander is the Oppenheim ``trickle-down method'' \cite{Opp18}, where one can show that if $\gamma_{\dimension-2}\ll \frac{1}{\dimension}$ and $X$ is ``well connected'' then all $\gamma_i$'s are at most $2\gamma_{\dimension-2}$.

In our main technical contribution, we give recursive framework for bounding the local spectral expansion of a weighted simplicial complex $X$ which significantly generalizes Oppenheim's result. 

\begin{theorem}[Inductive Matrix Trickle-Down Method]\label{thm:matrixtrickledowninductive}
Let $(X, \pi_\dimension)$ be a totally connected weighted simplicial complex. Suppose  $\{M_{\tau} \in \mathbb{R}^{X(0) \times X(0)}\}_{\tau \in X(\leq \dimension-2)}$ is a family of symmetric matrices satisfying the following:
\begin{enumerate}
    \item \textbf{Base Case:} For every $\tau$ of co-dimension $2$, we have the spectral inequality $$\Pi_{\tau}P_{\tau} - 2\pi_{\tau}\pi_{\tau}^{\top} \preceq M_{\tau} \preceq \frac{1}{5} \Pi_{\tau}.$$
    \item \textbf{Recursive Condition:} For every $\tau$ of co-dimension at least $k \geq 3$,  
    $M_{\tau}$ satisfies 
    $$\boundMatrix_\tau \preceq  \frac{k-1}{2k-1} \Pi_\tau \quad \text{and}\quad  \E_{x \sim \pi_{\tau}} M_{\tau \cup \{x\}} \preceq M_{\tau} - \frac{k-1}{k-2} M_{\tau}\Pi_{\tau}^{-1}M_{\tau}.$$
\end{enumerate}
Then  $\lambda_2(P_\tau)\leq \rho(\Pi_\tau^{-1}M_\tau)$ for all $\tau\in X(\leq d-2)$, where $\rho$ represents the spectral radius. In particular, $(X,\pi_\dimension)$ is a $(\mu_{-1},\dots,\mu_{\dimension-2})$-local spectral expander with $\mu_{k} = \max_{\tau \in X(k)} \rho(\Pi_\tau^{-1}M_\tau)$.
\end{theorem}

\begin{remark}
The original ``trickle-down method'' is a special case of our result by taking the matrix $M_{\tau }$ to be a multiple of $\Pi_{\tau}$. We justify this formally through a quick calculation in \cref{sec:matrixtrickledown} below.
\end{remark}
\begin{remark}
It also turns out for our applications to proper colorings we will need a slight extension of the above theorem. However, one should take the above theorem as the heart of our technical contributions. See \cref{thm:fulltechnical} below and the surrounding discussion for more details on the slight extension.
\end{remark}
Typically, it is not  difficult to construct some family of matrices $\{M_{\tau}\}$ satisfying the assumptions of \cref{thm:matrixtrickledowninductive}. The key challenge is choose $M_\tau$'s such that  one can bound $\rho(\Pi_\tau^{-1}M_{\tau})\leq O(1/\codimfunc(\tau))$. One of our second key insights is that the matrices $M_{\tau}$ can be designed to have convenient sparsity patterns depending on $(X,\pi_d)$ which allow for straightforward bounds on $\rho(\Pi_\tau^{-1}M_{\tau})$. For instance, in our application to proper colorings, our matrices $M_{\tau}$ will have rows and columns corresponding to vertex-color pairs $vc$, and they will be supported on the ``proper coloring constraints'', namely pairs $uc, vc'$ of vertex-color pairs where the vertices $u,v$ are neighbors and the two colors $c,c'$ are identical. We demonstrate the usefulness of this approach on sampling proper colorings in graphs below.

\subsection{Related Prior Work}
Sampling uniformly random proper colorings of bounded-degree graphs is a well-studied problem going back to the 1990s \cite{Jer95,SS97,Vig99} with applications to statistical physics. As alluded to before, nearly all prior results showing rapid mixing of the Glauber dynamics for this problem used variants of the coupling method or the correlation decay property.

Another method of attack for graph coloring is based on deterministic algorithms where one typically exploits Weitz's elegant algorithmic framework \cite{Wei04} based on the strong spatial mixing or the Barvinok's polynomial interpolation method \cite{Sok01,GK12,LY13,LSS19,BDPR21}.

Recently, it was observed  \cite{ALO20, CGSV21, FGYZ21} that the perspective of high-dimensional expanders, and in particular local spectral expansion, yields powerful new methods to obtain rapid mixing for multi-state spin systems. These results crucially relied on the correlation decay property  to bound the local spectral expansion of desired distribution. Specifically, for colorings, \cite{CGSV21, FGYZ21} showed that the correlation decay results of \cite{GKM15} give strong local spectral expansion bounds for proper colorings on triangle-free graphs when $q > \alpha \Delta$, where $\alpha \approx 1.763 < 2$ is a constant. They concluded rapid mixing of the Glauber dynamics in this regime, a result that seems difficult to obtain using coupling arguments.


However, despite the power of this approach, the main difficulty is that obtaining correlation decay for proper colorings is extremely challenging.
Along this line, our main technical theorem deviates from this recent trend as we directly bound the local spectral expansion of the underlying complex by induction instead of appealing to the correlation decay property.

\section{Preliminaries}
First, we fix some notational conventions. 
 When it is clear from context, we write $a$ to denote a singleton $\{a\}$.  When we want to add a subscript $b$ to an object denoted by $x_a$, we write $x_{a, b}$, i.e.  $x_{a, b} \coloneqq (x_a)_b $. We use the same convention for superscripts. Throughout the paper, adding or removing $\emptyset$ as a subscript  does not change the object.
 For an integer $q$,  we denote $\{1, \dots, q\}$ by $[q]$.  For a function $f: D \rightarrow R$ and a $D' \subseteq D$,  we write $f|_{D'}$  to denote the restriction of $f$ to domain $D'$. 
 \paragraph{Matrices and Vectors}
 Given a set $S$, we write $v \in \mathbb{R}^S$ and $A \in \mathbb{R}^{S \times S}$ to respectively  denote  a vector and a matrix indexed by $S$. We see a  probability  distribution $p$ over a set $S$ as a vector $p \in \mathbb{R}_{\geq 0}^S$.
  For a $n \times n$ matrix $A$, with eigenvalues $\lambda_1,\dots,\lambda_n$, we write $\rho(A)=\max_{1\leq i\leq n}|\lambda_i|$ to denote the spectral radius of $A$ and  $\norm{A}_\infty = \max_{1 \leq i \leq n } \{ \sum_{j=1}^n  |A_{ij}|\}$. We write $\preceq$ to denote the Loewner order, i.e. for any symmetric matrices $A, B \in \mathbb{R}^{S \times S}$, we write $A \preceq B$ if $B-A$ is positive semidefinite. For any matrix $A  \in \mathbb{R}^{S \times S}$ and any $S'\subseteq S$, $A^{S'} \in \mathbb{R}^{S \times S}$ is defined to be $A^{S'} (x, y) = A (x, y)$ for $x, y \in S'$, and $0$ on all other entries.  We say a matrix $A$ is diagonal if  $A(x, y) = 0$ whenever $x \neq y$. We say matrix $A \in \mathbb{R}^{S \times S}$ is hollow if $A(x, x) = 0$ for all $x \in S$.    For any matrix $A$ define  $A^+$ as  $A^+ (x, y) \coloneqq \max \{0, A(x, y)\}$ and $A^- \coloneqq A - A^+$. Furthermore, $\diag(A)$ is  a diagonal matrix whose diagonal   elements  match those of $A$, i.e. $A(x, x) = \diag (A)(x, x)$. We define $\offdiag(A) \coloneqq A - \diag(A)$. 
\paragraph{Graphs}
Fix a graph $G = (V, E)$. We denote the degree of each vertex $v$ by $\Delta_G(v)$ and the maximum degree of the graph by   $\Delta_G$. Let $n_G \coloneqq |V|$ and $m_G \coloneqq |E|$. Furthermore, let $G[U]$ be the induced subgraph of $G$ on $U \subseteq V$. We may drop the subscripts when it is clear from context. For any $u, v \in V$, we write $u \sim v$ if $\{u, v\} \in E$. Furthermore, for any $e \in E$, we write $e \sim v$ whenever vertex $v$ is an endpoint of $e$, and for  an edge $f \neq e$ we write $e \sim f$ when $e$ and $f$ share an endpoint. For an edge $e= \{u, v\}$, define $\Delta(e)$ as the number of edges that share an endpoint with $e$, i.e.  $\Delta(e) \coloneqq \Delta (u) + \Delta (v) - 2$.  The line graph $L(G)$  of a graph $G$ is defined as follows: every vertex in  $L(G)$ corresponds to an edge in $G$ and there is an edge between two vertices in $L(G)$ if their corresponding edges in $G$   share an endpoint. 

\subsection{Linear Algebra}
\begin{fact}\label{fact:rowsum}
For any symmetric matrix $A \in \R^{S\times S}$ where $A_{i,j}\neq 0$ only for $i,j\in S'\subseteq S$, we have  $A \preceq \|A \|_\infty I^{S'}$. 
\end{fact}
\begin{fact}\label{crosstermfact}
For  matrices $A,B\in \R^{m\times n}$ and positive $\eps >0 $, we have the inequalities $ AB^\top+BA^\top \preceq \eps AA^\top + \frac{1}{\eps} BB^\top $ and $(A+B)(A+B)^\top \preceq (1+\eps)AA^\top+(1+1/\eps)BB^\top$.
\end{fact}
\begin{proof}
Since $\eps>0$ we can write 
$$ 0\preceq \wrapp{\sqrt{\epsilon}A - \frac{1}{\sqrt{\epsilon}}B}\wrapp{\sqrt{\epsilon}A - \frac{1}{\sqrt{\epsilon}}B}^\top = \eps AA^\top + \frac{1}{\epsilon}BB^\top -AB^\top - BA^\top.$$
Rearranging yields the first inequality $AB^\top + BA^\top \preceq \epsilon AA^\top + \frac{1}{\epsilon}BB^\top$. Adding $AA^\top + BB^\top$ to both sides yields the second inequality.
\end{proof}


\begin{lemma}\label{lem:oppmonotone}
Let $A,B \in \R^{n \times n}$ be symmetric matrices such that $A \cdot(I- \alpha A) \preceq B \cdot(I - \alpha B)$ for a positive real number $\alpha > 0$. If $A, B \preceq \frac{1}{2\alpha} \cdot I$, then $A \preceq B$. Note that we crucially do not require $A, B \succeq 0$.
\end{lemma}
\begin{proof}
It suffices to prove the claim when $\alpha = 1$, since the general claim then follows by replacing $A,B$ with $\alpha A, \alpha B$, respectively.

First, observe the matrix map $M \mapsto M(I - M)$ is a bijection between $\wrapc{M \in \R^{n \times n}: M \preceq \frac{1}{2} I}$ and $\wrapc{T \in \R^{n \times n}: T \preceq \frac{1}{4} I}$ with inverse
\begin{align}\label{eq:AIAmapinverse}
    T \mapsto \frac{1}{2} I - \wrapp{ \frac{1}{4} I - T}^{1/2}.
\end{align}
The way to see this is via the eigendecomposition: if $M = \sum_{i=1}^{n} \mu_{i} \varphi_{i} \varphi_{i}^{\top}$ for an orthonormal eigendecomposition $\{\varphi_{i}\}_{i=1}^{n}$ with corresponding eigenvalues $\{\mu_{i}\}_{i=1}^{n}$, then $M(I - M)=\sum_{i=1}^{n} \mu_{i}(1 - \mu_{i}) \varphi_{i} \varphi_{i}^{\top}$. Hence, to prove this claim, it suffices to show that the real function $x \mapsto x(1 - x)$ is a bijection between $(-\infty,1/2]$ and $(-\infty,1/4]$. To see this, observe that the quadratic $x(1-x)=\mu$ has roots $x = \frac{1}{2} \pm \wrapp{\frac{1}{4} - \mu}^{1/2}$, and since we enforced that $\mu \leq \frac{1}{4},$ we must choose $x$ to be the smaller root, i.e. $\frac{1}{2} - \wrapp{\frac{1}{4} - \mu}^{1/2} \leq \frac{1}{2}$ gives the inverse function.

Knowing this explicit inverse function, we now return to the proof of the lemma. Since $A,B \preceq \frac{1}{2}I$, we may apply the inverse \cref{eq:AIAmapinverse} to $A(I-A)$ (resp. $B(I-B)$) to recover $A$ (resp. $B$). Hence, to prove the claim, it suffices to establish operator monotonicity of \cref{eq:AIAmapinverse}. A quick calculation reveals that this is equivalent to operator monotonicity of $M \mapsto \sqrt{M}$ for positive semidefinite $M \in \R^{n \times n}$, which is well-known and follows for instance by using the L\"{o}wner-Heinz Theorem \cite{Lowner34}.
\end{proof}

\subsection{Markov Chains}\label{subsec:markovchain}
Let $P$ be the transition probability matrix of a Markov chain on a finite state space $\Omega$ with stationary distribution $\pi$. We say $P$ is irreducible if $P$ is connected. We say $P$ is reversible w.r.t. $\pi$ if for all $x,y \in \Omega$, we have $\pi(x)P(x, y) = \pi(y)P(y, x)$. In this case, the matrix $P$ becomes self-adjoint w.r.t. the natural inner product $\langle\cdot,\cdot\rangle_{\pi}$ induced by $\pi$ on $\R^{\Omega}$ given by $\langle \phi, \psi \rangle_{\pi} = \E_{\pi}[\phi\psi]$. Throughout, we only work with irreducible reversible Markov chains.

We will be interested in the mixing of our Markov chains, which quantifies the rate of convergence to stationarity. Specifically, for an initial starting distribution $\mu$ on $\Omega$ and error parameter $\epsilon > 0$, define
\begin{align*}
    t_{\mix}(P,\epsilon, \mu) \overset{\defin}{=} \min\{t \geq 0 : \norm{\mu P^{t} - \pi}_{\TV} \leq \epsilon\},
\end{align*}
 where $\norm{\mu - \nu}_{\TV} = \frac{1}{2} \sum_{x \in \Omega} \abs{\mu(x) - \nu(x)}$ gives the total variation distance between two distributions $\mu,\nu$ on $\Omega$. We write $t_{\mix}(P,\epsilon) = \sup_{\mu} t_{\mix}(P, \epsilon, \mu)$, where the supremum is over all possible starting distributions $\mu$. The mixing time of $P$ is defined as $t_{\mix}(P) = t_{\mix}(P, 1/4)$.

It is well-known that the mixing time is controlled by various constants arising from classical functional inequalities. To introduce these, we write $\mathcal{E}_{P}(\phi,\psi) = \langle \phi, (I - P)\psi \rangle_{\pi}$ for the Dirichlet form of $P$, $\var_{\pi}(\phi) = \E_{\pi}[\phi^{2}] - \E_{\pi}[\phi]^{2}$ for the variance of $\phi$ w.r.t. $\pi$, and $\Ent_{\pi}(\phi) = \E_{\pi}[\phi \log \phi] - \E_{\pi}[\phi]\log \E_{\pi}[\phi]$. With these in hand, we define
\begin{enumerate}
    \item Spectral Gap: $\lambda(P) \overset{\defin}{=} \inf_{f \neq 0} \frac{\mathcal{E}(f, f)}{\var_{\pi}(f)}$
    \item Modified Log-Sobolev Constant: $\rho(P) \overset{\defin}{=} \inf_{f \geq 0} \frac{\mathcal{E}(f, \log f)}{\Ent_{\pi}(f)}$
    \item Standard Log-Sobolev Constant: $\kappa(P) \overset{\defin}{=} \inf_{f \geq 0} \frac{\mathcal{E}(\sqrt{f}, \sqrt{f})}{\Ent_{\pi}(f)}$
\end{enumerate}
Bounds on these constants have the following consequences for the mixing time.
\begin{proposition}
For an irreducible reversible Markov chain $P$ on a finite state space $\Omega$ with stationary distribution $\pi$, we have the following inequalities:
\begin{align*}
    t_{\mix}(P) &\leq O\wrapp{\frac{1}{\lambda(P)} \log \frac{1}{\pi_{\min}}} \tag{\cite{LPW17}} \\
    t_{\mix}(P) &\leq O\wrapp{\frac{1}{\rho(P)} \log \log \frac{1}{\pi_{\min}}} \tag{\cite{BT03}} \\
    t_{\mix}(P) &\leq O\wrapp{\frac{1}{\kappa(P)} \log \log \frac{1}{\pi_{\min}}} \tag{\cite{DS96}}
\end{align*}
\end{proposition}
We note that bounds on the standard and modified log-Sobolev constants also yield sub-Gaussian concentration estimates \cite{Goe04, Sam05, BLM16}, but we will not need these in our applications.

\subsection{Products of Weighted Simplicial Complexes}
Given two pure simplicial complexes $X,Y$ of dimensions $\dimension_{X},\dimension_{Y}$ respectively with disjoint ground sets, we may form another pure simplicial complex $Z$ of dimension-$d_Z$ for $\dimension_Z \coloneqq \dimension_X+ \dimension_Y+1$ called the product $X \times Y$ of $X,Y$ by taking the ground set of $X \times Y$ to be the disjoint union of the ground sets of $X,Y$, taking the facets of $X \times Y$ to be of the form $\tau \cup \sigma$, where $\tau,\sigma$ are facets of $X,Y$ respectively, and then taking downwards closure. If $\pi_{X, \dimension_X},\pi_{Y, \dimension_Y}$ are distributions on the facets of $X,Y$ respectively, we then form corresponding product distribution $\pi_{Z, \dimension_Z} = \pi_{X, \dimension_X} \times \pi_{Y, \dimension_Y}$ on the facets of $X \times Y$ by taking $\pi_{Z, \dimension_Z}(\tau \cup \sigma) = \pi_{X, \dimension_X}(\tau)\cdot\pi_{Y, \dimension_Y}(\sigma)$ for facets $\tau,\sigma$ of $X,Y$ respectively.

\begin{lemma}\label{lem:productlocalwalks}
Given a  weighted simplicial complex $(Z, \pi_{Z, \dimension_Z}) = (X \times Y, \pi_{X, \dimension_X} \times \pi_{Y, \dimension_Y})$, we may decompose $P_{Z,\emptyset}$ into $P_{X,\emptyset},P_{Y,\emptyset}$ via the formula
\begin{align*}
    P_{Z} - \frac{\dimension_Z+1}{\dimension_Z}\mathbf{1}\pi_{Z}^{\top} = \begin{bmatrix} \frac{\dimension_X}{\dimension_{Z} }\wrapp{P_{X} - \frac{\dimension_{X}+1}{\dimension_{X}} \mathbf{1}\pi_{X}^{\top}} & 0 \\ 0 & \frac{\dimension_Y}{\dimension_{Z} }\wrapp{P_{Y} - \frac{\dimension_{Y}+1}{\dimension_{Y}} \mathbf{1} \pi_{Y}^{\top}} \end{bmatrix}
\end{align*}
where we put a zero matrix  for $P_X$ if $d_X=0$ (and similarly for $P_Y$).
In particular, if $M_{X},M_{Y}$ are symmetric matrices satisfying
$\Pi_{X} P_{X} - \frac{\dimension_{X}+1}{\dimension_{X} }\pi_{X}\pi_{X}^{\top} \preceq M_{X}$  and $\Pi_{Y} P_{Y} - \frac{\dimension_{Y}+1}{\dimension_{Y}}\pi_{Y}\pi_{Y}^{\top} \preceq M_{Y}$ (we let $M_X = 0$ if $\dimension_X= 0$, and similarly $M_Y = 0$ if $\dimension_Y = 0$), then
\begin{align*}
    \Pi_{Z} P_{Z} - \frac{\dimension_Z+1}{\dimension_Z} \pi_{Z}\pi_{Z}^{\top} \preceq \begin{bmatrix} \frac{\dimension_{X}(\dimension_{X}+1) }{\dimension_Z(\dimension_{Z}+1)}M_{X} & 0 \\ 0 & \frac{\dimension_Y(\dimension_Y+1)}{\dimension_Z(\dimension_Z+1)} M_{Y} \end{bmatrix}
\end{align*}
\end{lemma}
\begin{proof}
Let $x \in X(0)$ and $y \in Y(0)$. 
Then, 
    \begin{align*}
           P_Z(x, y) \underset{\cref{eq:localwalkdef}}{=} \frac{1}{d_Z} \Pr_{\sigma \sim \pi_{Z, \dimension_Z}}[y \in \sigma \mid x \in \sigma ] 
         \underset{\text{By independence}}{=} \frac{1}{d_Z} \Pr_{\sigma \sim \pi_{Z, \dimension_Z}}[y \in \sigma]   \underset{\cref{eq:pidef}}{=} \frac{d_Z+1}{d_Z} \pi_Z(y).
    \end{align*}
    Therefore $(P_{Z} - \frac{\dimension_Z+1}{\dimension_Z}\mathbf{1}\pi_{Z}^{\top}) (x, y) = 0 $ for all $x \in X(0) , y \in Y(0)$ ( similarly $(P_Z - \frac{\dimension_Z+1}{\dimension_Z}\mathbf{1}\pi_{Z}^{\top}) (y, x) = 0 $). For $x,x' \in X(0)$, if $x \neq x'$, by   \cref{eq:localwalkdef} we have
      \begin{align*}
           P_Z(x, x')& = \frac{1}{d_Z} \Pr_{\sigma \sim \pi_{Z, \dimension_Z}}[x' \in \sigma \mid x \in \sigma ]  
         = \frac{1}{d_Z} \Pr_{\sigma \sim \pi_{X, \dimension_X}}[x' \in \sigma | x\in \sigma] =  \frac{d_X}{d_Z}  P_X(x, x').
    \end{align*}
If $x = x'$, then $P_Z(x, x) =  \frac{d_X}{d_Z}  P_X(x, x) = 0 $ by definition of the local random walk. Note that, if $d_X = 0$,  $P_X(x, x')$ is defined to be  0 by  the statement of   the theorem. Furthermore, for $x,x' \in X(0)$,  by  \cref{eq:pidef}
\begin{align*}
          \mathbf{1} \pi_Z^\top(x, x') = \frac{1}{d_Z+1} \Pr_{\sigma \sim \pi_{Z, \dimension_Z}}[x' \in \sigma]  
         = \frac{1}{d_Z+1} \Pr_{\sigma \sim \pi_{X, \dimension_X}}[x' \in \sigma] =  \frac{d_X +1}{d_Z+1}  \mathbf{1} \pi_Z^\top (x, x').
    \end{align*}
Putting these together, 
$ (P_{Z} - \frac{\dimension_Z+1}{\dimension_Z} \pi_{Z}\pi_{Z}^{\top}) (x, x') = \left( \frac{\dimension_X}{\dimension_{Z} }(P_{X} - \frac{\dimension_{X}+1}{\dimension_{X}} \mathbf{1}\pi_{X}^{\top})\right)(x, x')$ for all $x, x' \in X(0)$. One can see that a similar statement holds for the lower-right block.  This finishes the proof of the first part. To get the second part, it is enough to note that when $d_X = 0$, on the $X(0) \times X(0)$ block of  $P_{Z} - \frac{\dimension_Z+1}{\dimension_Z} \mathbf{1} \pi_{Z}^{\top}$ we have  $- \frac{\dimension_X+1}{\dimension_Z}  \mathbf{1}\pi_{X}^{\top} \preceq 0$.
\end{proof}
\par A natural example of product of simplicial complexes is the vertex-coloring complex of a disconnected graph. Say $G=(V,E)$ is a graph with $n$ vertex and  $\ell$ connected components  $G[U_1],\dots,G[U_\ell]$ and associated  complexes  $X_{1},\dots X_{\ell}$.  If $(X, \pi_{n-1})$ is the complex associated with $G$ and $\pi_{n-1}$ is the uniform distribution over its facets, then we can write $(X, \pi_{n-1}) =(X_{1}, \mu_1) \times\dots\times (X_{l}, \mu_\ell)$, where $\mu_{i}$  is the uniform distribution over facets of $X_i$.
Suppose we have associated a matrix $A_\tau\in \R^{X(0)\times X(0)}$ to any non-empty face $\tau$  of co-dimension at least 2 and assume that for any $1\le i \leq \ell$, when  $\tau_{-i}$ and $\sigma_{-i}$ are  two arbitrary  colorings of all connected components except $i$, then   $A_{\tau_{-i}} = A_{\sigma_{-i}}$.  We associate a block-diagonal matrix 
$$f_{\times}(X, \{A_\tau\}_{\emptyset\subsetneq \tau\in X(\leq n-3)}):=\sum_{1 \leq i\leq \ell: |U_i| \neq 1} A_{\tau_{-i}}.$$
When $X$ is the edge-coloring complex of a graph $G$, it is the vertex-coloring complex of the line graph of $G$, therefore  $f_{\times}(X, \{A_\tau\}_{\emptyset\subsetneq \tau\in X(\leq n-3)})$ is given by the above definition.  

\subsection{Garland's Method}
We will need the following simple facts, which follow simply by applying the Law of Total Probability appropriately and using the definition of the local walks $P$ and local distributions $\pi$. Nearly identical equations were first observed and found to be useful by Garland \cite{Gar73} in the context of understanding cohomology of simplicial complexes. They also lie at the heart of understanding expansion phenomena, in particular local spectral expansion, in simplicial complexes \cite{Opp18, KO18}.
\begin{lemma}[\cite{Opp18}]\label{lem:PiLocalization}
Given  a weighted simplicial complex $(X, \pi_\dimension)$, we may decompose $\Pi P$ as
\begin{align*}
    \Pi P = \E_{x \sim \pi} \Pi_{x} P_{x}
\end{align*}
\end{lemma}
\begin{proof}
We check entry by entry, applying the Law of Total Probability along the way. Note that the matrices on both sides have zero diagonal and so it suffices to check equality for the $(y,z)$-entry, where $y \neq z$.
\begin{align*}
    \E_{x \sim \pi} (\Pi_{x}P_{x})(y,z) &= \sum_{x \neq y,z} \pi(x) \Pi_{x}(y) P_{x}(y, z) \\
    &=\sum_{x\neq y,z}\cdot \frac{1}{\dimension
    +1}\Pr_{\sigma \sim \pi_{\dimension}}[x \in \sigma]\cdot \pi_x(y)\cdot \frac{\pi_x(\{y,z\})}{2\pi_x(y)}\tag{\cref{eq:pidef} and \cref{eq:localwalkdef}}\\
    &= \sum_{x \neq y,z} \frac{1}{\dimension
    +1}\Pr_{\sigma \sim \pi_{\dimension}}[x \in \sigma] \cdot \frac{1}{\dimension(\dimension-1)}\Pr_{\sigma \sim \pi_{\dimension}}[y,z \in \sigma \mid x \in \sigma] \tag{\cref{eq:pidef}} \\
    &= \frac{1}{\dimension(\dimension+1)} \cdot \frac{1}{\dimension-1} \sum_{x \neq y,z} \Pr_{\sigma \sim \pi_{\dimension}}[x,y,z \in \sigma] \tag{Rearranging} \\
    &= \frac{1}{\dimension(\dimension+1)} \Pr_{\sigma \sim \pi_{\dimension}}[y,z \in \sigma] \tag{Law of Total Probability} \\
    &= \pi(y)P(y, z) \tag{\cref{eq:pidef} and \cref{eq:localwalkdef}} \\
    &= (\Pi P)(y,z)
\end{align*}
\end{proof}
\begin{lemma}[\cite{Opp18}]\label{lem:OppenheimLoss}
Given a weighted simplicial complex $(X, \pi_\dimension)$, we may decompose $\Pi P^{2}$ as
\begin{align*}
    \Pi P^{2} = \E_{x \sim \pi} \pi_{x}\pi_{x}^{\top}
\end{align*}
\begin{proof}
The main observation is that the rows of $P$ are precisely the vectors $\pi_{x}$, and that the rows of $\Pi P$ are precisely the vectors $\pi(x)\pi_{x}$. The claim immediately follows.
\end{proof}
\end{lemma}

\subsection{Local-to-Global Theorems}\label{subsec:localtoglobal}
As alluded to earlier, one of the beautiful insights of \cite{DK17,KO18} is that local spectral expansion implies quantitative bounds on the spectral gap of the down-up walk.
\begin{theorem}[\cite{AL20}]\label{thm:localtoglobal}
Let $(X,\pi_{\dimension})$ be a $(\gamma_{-1},\dots,\gamma_{\dimension-2})$-local spectral expander. Then the down-up walk which samples from $\pi_{d}$ has spectral gap lower bounded by
\begin{align*}
    \lambda(P_{\dimension}^{\vee}) \geq \frac{1}{\dimension} \prod_{j=-1}^{\dimension-2} (1 - \gamma_{j})
\end{align*}
\end{theorem}
Analogous results have also been proved for the decay of entropy \cite{CLV21, GM20, AASV21}. We state this result here since our main technical result \cref{thm:matrixtrickledowninductive} is a general method to obtain local spectral expansion for any weighted simplicial complex.

\cref{thm:localtoglobal}\label{subsec:local-to-global} has been used successfully in several prior works \cite{ALOV19ii, ALO20, CLV20, CGSV21, FGYZ21,  ALOVV21}  to establish polynomial time mixing for various dynamics. However, the exponent of these running times are typically large, depending sensitively on the constants in the $\gamma_{j}$. Recently, it was shown that for weighted simplicial complexes arising from spin systems on bounded-degree graphs, we can do significantly better. We will need the following to establish optimal mixing times in our applications. We state it specifically for colorings, as we do not analyze any other spin systems in this paper.
\begin{theorem}[\cite{CLV21, BCCPSV21}]\label{thm:boundeddegspinlocaltoglobal}
Let $(X,\pi_{n-1})$ be a weighted simplicial complex arising from the uniform distribution over proper list-colorings of a graph $G=(V,E)$ with $|V| = n$ and maximum degree $\Delta$. If for some constant $C$ independent of $n$, $(X,\pi_{n-1})$ is a $(\gamma_{-1},\dots,\gamma_{n-3})$-local spectral expander with $\gamma_{k} \leq \frac{C}{n-k-1}$ for all $k$, then the spectral gap, standard and modified log-Sobolev constants are all $\Omega_{C,\Delta}(1/n)$. In particular, the Glauber dynamics mixes in $O_{C,\Delta}(n\log n)$ steps.
\end{theorem}



\section{A General Matrix Trickle-Down Method}\label{sec:matrixtrickledown}
Our goal in this section is to prove \cref{thm:matrixtrickledowninductive} and its extension \cref{thm:fulltechnical}. Towards this, we first elucidate the original ``trickle-down method'' of Oppenheim \cite{Opp18}, which was the beautiful realization that one can bound the second eigenvalue of the local walk $P$ in terms of the second eigenvalues of the local walks $\{P_{x}\}_{x \in X(0)}$. This ``trickle-down'' phenomenon naturally yields an inductive method of bounding the second eigenvalues of all local walks. We formalize this as follows.
\begin{theorem}[Trickle Down Theorem \cite{Opp18}]\label{thm:vanillatrickledown}
Given a weighted simplicial complex $(X, \pi_\dimension)$, suppose the following holds:
\begin{enumerate}
    \item \textbf{Connectivity:} $\lambda_{2}(P) < 1$, i.e. the local walk $P$ is connected/irreducible.
    \item \textbf{Spectral Bound for Links Above:} For some $0 \leq \lambda \leq 1/2$, we have the bound $\lambda_{2}(P_{x}) \leq \lambda$ for all $i \in X(0)$.
\end{enumerate}
Then the local walk $P$ actually satisfies the spectral bound $\lambda_{2}(P) \leq \frac{\lambda}{1 - \lambda}$.
\end{theorem}
\begin{remark}
The original statement in \cite{Opp18} does not have the assumption $\lambda \leq 1/2$, but the two are completely equivalent, since if $\lambda > 1/2$, then $\frac{\lambda}{1 - \lambda} > 1$, making the statement vacuously true.
\end{remark}
In particular, suppose we know that the top-dimensional links have local walks with second eigenvalue upper bounded by $1/(\dimension+1)$. Then by applying \cref{thm:vanillatrickledown} in a totally blackbox fashion, it immediately follows that $(X,\pi_{\dimension})$ is a $\wrapp{\frac{1}{2},\dots,\frac{1}{\dimension+1}}$-local spectral expander and its high-order walk has spectral gap at least $\Omega(1/d^{2})$. This result has already had many applications, such as the construction of bounded-degree high-dimensional expanders \cite{KO18b} and sampling algorithms for matroids \cite{ALOV19ii, CGM19, ALOVV21} and other combinatorial structures \cite{AL20}.

Unfortunately, when $X$ is the simplicial complex of proper (partial) colorings, and $\pi_{n-1}$ is the uniform distribution over proper colorings, this standard trickle-down method is not enough. It turns out that in the worst case, the second largest eigenvalue of the local walk of a top-dimensional link is $\Theta\wrapp{\frac{1}{q - \Delta}} = \Theta(1)$, which is much too large since \cref{thm:vanillatrickledown} needs to be applied $n-2$ times.

Our main technical contribution is to provide a framework which significantly generalizes \cref{thm:vanillatrickledown}, and makes the trickle-down theorem applicable to wider classes of weighted simplicial complexes, including those arising from proper colorings. One of our main insights is to replace the hypothesis that $\lambda_{2}(P_{x}) \leq \lambda$, which merely provides a uniform bound on all nontrivial eigenvalues of $P_{x}$, by a matrix bound ``$P_{x} \preceq M_{x}$''. The hope is that the matrix $M_{x}$ itself can simultaneously be easily bounded, as well as provide information on where the ``bad'' eigen-spaces of $P_{x}$ are. 
So, roughly speaking, although many of the top dimensional links $P_\tau$'s may have a constant second eigenvalue, by carefully choosing $M_\tau$'s one can ``average out'' these bad eigen-spaces to show that the link of a lower dimensional face has small eigenvalues.
We formalize this as follows.

\begin{theorem}[Matrix Trickle-Down Theorem]\label{thm:trickledowngeneral}
Given  a weighted simplicial complex $(X, \pi_\dimension)$, suppose the following holds:
\begin{enumerate}
    \item \textbf{Connectivity:} $\lambda_{2}(P) < 1$, i.e. the local walk $P$ is connected/irreducible.
    \item \textbf{Generalized Spectral Bound for Links Above:} There is a family of symmetric matrices $\{M_{x}\}_{x \in X(0)}$ such that
    \begin{align*}
        \Pi_{x}P_{x} - \alpha \pi_{x}\pi_{x}^{\top} \preceq M_{x} \preceq \frac{1}{2\alpha + 1} \Pi_{x}
    \end{align*}
    for all $i \in X(0)$.
\end{enumerate}
Then the local walk $P$ actually satisfies the spectral bound $\Pi P - \wrapp{2 - \frac{1}{\alpha}} \pi\pi^{\top} \preceq M$, and in particular $\lambda_{2}(P) \leq \rho(\Pi^{-1}M)$, where $M$ is any symmetric matrix satisfying $M \preceq \frac{1}{2\alpha}\Pi$ and 
$\E_{x \sim \pi} M_{x} \preceq M - \alpha M\Pi^{-1}M$. 
\end{theorem}
Note that by induction, \cref{thm:matrixtrickledowninductive} follows as an immediate consequence of this generalized trickle-down result.

To see that the above theorem generalizes \cref{thm:vanillatrickledown}, note that if $\lambda_{2}(P_{x})\leq \lambda\leq 1/2$ for all $i\in X(0)$ then $M_{x}=\lambda \Pi_{x}$, $\alpha=1-\lambda$, and $M=\frac{\lambda}{1-\lambda} \Pi$ satisfies the assumptions of the above theorem; in particular,
\begin{align*}
M_{x} &=\lambda \Pi_{x} \leq \frac{1}{2(1-\lambda)+1}\Pi_{x}=\frac{1}{2\alpha+1}\Pi_{x}\tag{$\lambda\leq 1/2$} \\
\Pi_{x} P_{x} &- (1-\lambda)\pi_{x}\pi_{x}^{\top} \preceq \lambda \Pi_{x} = M_{x} \tag{$\lambda_2(P_x)\leq \lambda$ and $\pi_{x}P_{x} = \pi_{x}$} \\
M &=\frac{\lambda}{1-\lambda}\Pi \preceq \frac{1}{2(1-\lambda)}\Pi=\frac{1}{2\alpha}\Pi \tag{$\lambda<1/2$}\\
\mathbb{E}_{x\sim\pi} M_{x} &= \lambda \E_{x \sim \pi} \Pi_{x} = \lambda \Pi  \\
&= \frac{\lambda}{1-\lambda} \Pi - (1-\lambda)\wrapp{\frac{\lambda}{1-\lambda}}^{2}\Pi \\
&= M - \alpha M\Pi^{-1}M.
\end{align*}
so, assuming connectivity of $P$, we get $\lambda_{2}(P)\leq \rho(\Pi^{-1}M)=\frac{\lambda}{1 - \lambda}$ from \cref{thm:trickledowngeneral} as desired.

Let us now prove \cref{thm:trickledowngeneral}. To do this, we first need the following lemma.
\begin{lemma}\label{lem:opploss2}
Let $(X, \pi_\dimension)$ be  a weighted simplicial complex. Suppose  for a symmetric matrix $M$ and an $\alpha \geq 1/2$, the matrix inequalities $M, \Pi P - \wrapp{2 - \frac{1}{\alpha}} \pi\pi^{\top} \preceq \frac{1}{2\alpha}\Pi$ hold and
\begin{align}\label{eq:opplossassumption}
    \Pi P - \alpha \Pi P^{2} \preceq M - \alpha M\Pi^{-1}M
\end{align}
Then we have the bound $\Pi P - \wrapp{2 - \frac{1}{\alpha}} \pi\pi^{\top} \preceq M$.
\end{lemma}
\begin{proof}
Our goal is to apply \cref{lem:oppmonotone} to suitably chosen $A,B$. Define $Q = P - \wrapp{2 - \frac{1}{\alpha}} \mathbf{1} \pi^{\top}$. A quick calculation shows that $Q - \alpha Q^{2} = P - \alpha P^{2}$, and so by multiplying both sides of \cref{eq:opplossassumption} by $\Pi^{-1/2}$, we see that \cref{eq:opplossassumption} is equivalent to
\begin{align*}
    \Pi^{1/2} Q \Pi^{-1/2} - \alpha \Pi^{1/2} Q^{2} \Pi^{-1/2} \preceq \Pi^{-1/2}M\Pi^{-1/2} - \alpha \Pi^{-1/2}M\Pi^{-1}M \Pi^{-1/2}
\end{align*}
Taking $A = \Pi^{1/2}Q \Pi^{-1/2}$ and $B = \Pi^{-1/2}M\Pi^{-1/2}$, we see by assumption that $A,B$ are symmetric matrices satisfying $A,B \preceq \frac{1}{2\alpha} I$ and $A(I - \alpha A) \preceq B(I - \alpha B)$. It follows by \cref{lem:oppmonotone} that $A \preceq B$, which is equivalent to $\Pi P - \wrapp{2 - \frac{1}{\alpha}}\pi\pi^{\top} = \Pi Q \preceq M$ as desired.
\end{proof}
With this lemma in hand, let us now prove \cref{thm:trickledowngeneral}.
\begin{proof}[Proof of \cref{thm:trickledowngeneral}]
The conclusion follows immediately from \cref{lem:opploss2}, and so it suffices to verify the conditions of the lemma. By assumption, we already have $M \preceq \frac{1}{2\alpha} \Pi$. Furthermore, $\Pi_{x}P_{x} - \alpha \pi_{x}\pi_{x}^{\top} \preceq M_{x} \preceq \frac{1}{2\alpha + 1} \Pi_{x}$ implies that $\lambda_{2}(P_{x}) \leq \frac{1}{2\alpha + 1}$. Since $\lambda_{2}(P) < 1$, by \cref{thm:vanillatrickledown} (the original Trickle-Down Theorem), $\lambda_{2}(P) \leq \frac{1}{2\alpha}$. Combined with the inequality $2 - \frac{1}{\alpha} \geq 1 - \frac{1}{2\alpha}$, which holds since $\alpha \geq 1/2$, it follows that $\Pi P - \wrapp{2 - \frac{1}{\alpha}}\pi\pi^{\top} \preceq \frac{1}{2\alpha}\Pi$.

All that remains is to verify \cref{eq:opplossassumption}. Observe that
\begin{align*}
    \Pi P &= \E_{x \sim \pi} \Pi_{x} P_{x} \tag{\cref{lem:PiLocalization}} \\
    &\preceq \E_{x \sim \pi}\wrapb{\alpha \pi_{x}\pi_{x}^{\top} + M_{x}} \tag{Assumption} \\
    &= \alpha \Pi P^{2} + \E_{x \sim \pi} M_{x} \tag{\cref{lem:OppenheimLoss}} \\
    &\preceq \alpha \Pi P^{2} + M - \alpha M\Pi^{-1}M \tag{Assumption}
\end{align*}
Rearranging, we obtain that $\Pi P - \alpha \Pi P^{2} \preceq M - \alpha M \Pi^{-1} M$ as desired.
\end{proof}

\subsection{A Slight Extension of \cref{thm:matrixtrickledowninductive}}
Here, we prove an extension of the matrix trickle-down method to take into account when simplicial complexes factor as products of smaller complexes. This will be useful in the context of proper colorings when the input graph is broken into several connect components by coloring some of the vertices.
\begin{theorem}\label{thm:fulltechnical}
Let $(X, \pi_{\dimension})$ be a totally connected weighted complex. Suppose $\{M_{\tau} \in \mathbb{R}^{X(0) \times X(0)}\}_{\tau \in X(\leq \dimension-2)}$ is a family of symmetric matrices satisfying the following:
\begin{enumerate}
    \item \textbf{Base Case:} For every $\tau$ of co-dimension 2, we have the spectral inequality $$\Pi_{\tau}P_{\tau} - 2\pi_{\tau}\pi_{\tau}^{\top} \preceq M_{\tau} \preceq \frac{1}{5} \Pi_{\tau}.$$
    \item \textbf{Recursive Condition:} For every $\tau$ of co-dimension at least $k \geq 3$, at least one of the following holds:  $ M_{\tau}$ satisfies 
   \begin{align}    \label{eq:desired_ineq}
 \boundMatrix_\tau \preceq  \frac{k-1}{3k -1} \Pi_\tau \quad \text{and}  \quad
  \E_{x \sim \pi_{\tau}} M_{\tau \cup \{x\}} \preceq M_{\tau} - \frac{k-1}{k-2} M_{\tau}\Pi_{\tau}^{-1}M_{\tau}.
   \end{align}
    Or,  $(X_\tau,\pi_{\tau, k-1})$ is a product of weighted simplicial complexes $(Y_{1},\mu_{1}),\dots,(Y_{t},\mu_{t})$ and  for every $\eta \in X_{\tau}(k-1)$,
        \begin{align*}
            M_{\tau} = \bigoplus_{1 \leq i \leq t : \dimension_{Y_i} \geq 1} \frac{\dimension_{Y_i}(\dimension_{Y_{i}}+1) }{k(k-1)} M_{\tau \cup \eta_{-i}},
        \end{align*}
        where $\eta_{-i} = \eta \setminus Y_{i}(0)$.

\end{enumerate}
Then for every $\tau \in X(\leq \dimension-2)$, we have the bound 
$\lambda_2(\Pi_\tau P_\tau)\leq \rho(\Pi_\tau^{-1}M_\tau)$.
\end{theorem}

\begin{proof}
Apply \cref{thm:trickledowngeneral} and \cref{lem:productlocalwalks} inductively.
\end{proof}

\section{Vertex Coloring}
\label{sec:vertexcoloring}

Fix an integer $q$  and graph  $G = (V, E)$ and a function $L$ that maps each $v \in V$ to a subset of $[q]$. We call $(G, L)$ a vertex-list-coloring instance. 
For any $u, v  \in V$, we write  $u \sim_c v$ when $u \sim v$  and $c \in L(u) \cap L(v)$. 
Furthermore, we define a $\beta$-extra color vertex-list-coloring instance as follows. 
\begin{definition}
We say a vertex-list-coloring instance $(G, L)$ is a $\beta$-extra-color instance if  for each $v \in V$, $|L(v)| \geq \beta + \Delta_G(v) $. 
\end{definition}

 We call an assignment $\sigma : V \rightarrow [q]$ a $L$-vertex-list-coloring of $G$ if $\sigma(v) \in L(v)$ for all $v \in V$; we say $\sigma$ is proper if $\sigma(u) \neq \sigma(v)$ whenever $u \sim v$. 
When it is clear from context we say $\sigma$ is a proper coloring to mean it is a proper $L$-vertex-list-coloring.  We say $\tau$ is   proper partial coloring on $U \subset V$ when it is a proper $L|_U$-vertex-list-coloring for $G[U]$.
 We may view list-colorings $\sigma$ as sets of vertex-color pairs $(v,c)$, which we denote by $vc$ for convenience.  When the graph $G$ and color lists $L$ are clear from context, we write $\pi_{n-1}$ for the uniform distribution over proper $L$-vertex-list-coloring of $G=(V,E)$. For a proper partial coloring on $U \subset V$, $v \in V \setminus U$ and $c \in L(v)$, define
\begin{align*}
    p(vc | \tau) \coloneqq \mathbb{P}_{\sigma \sim \pi_{n-1}} (\sigma(v)=c |\forall u \in U: \sigma (u) = \tau(u) ).
\end{align*}
In order to analyze the Glauber dynamics for an vertex-list-coloring instance $(G, L)$, we can build an $(n-1)$-dimensional simplicial complex such that the down-up walk on its facets is the same as the Glauber dynamics on $(G, L)$. Our aim is to  
apply  \cref{thm:fulltechnical} to bound the second eigenvalue of the transition probability matrix of the local walks and then use \cref{thm:localtoglobal} to bound the second eigenvalue of the  transition probability matrix of the global down-up walk on the facets. 
\begin{definition}[Simplicial Complex of a Vertex-List-Coloring Instance]
Given a vertex-list-coloring instance $(G, L)$, let $X(G, L)$ be a pure $(n -1)$-dimensional simplicial complex specified by the following facets:  $\{(v, \sigma(v))\}_{v \in V} $ is a facet  if and only if $\sigma$ is a proper $L$-vertex-list-coloring for $G$. 
\end{definition}
When it is clear from context, we abbreviate $  X(G, L)$ to $X$. 
Note that for all $0 \leq k \leq n$, any  face  $\tau$ of co-dimension $k$  is a proper partial coloring on a subset of vertices $U$ of size $n-k$, i.e., $k$ vertices remain uncolored.   Furthermore,  $X_\tau (0)$ can be seen as the set of all $vc$ such that $c \in L(v)$, $v \notin U$ and for any $u \sim v$, $uc \notin \tau$. 
We define $V_\tau \coloneqq \{ v : \exists c, vc\in X_\tau(0)\}$. Let $G_\tau \coloneqq G[V_\tau]$ and $\Delta_\tau(.)$ be the degree function of $G_\tau$ and $\Delta_\tau$ be its maximum degree. We define $L_\tau (v) \coloneqq \{ c \in L(v) : vc \in X_\tau (0) \}$ to be the list of remaining colors available to $v$ after coloring vertex $u$ with $c$ for each $uc \in \tau$. Let   $l_\tau (v) \coloneqq |L_\tau(v)|$ and  $l_\tau(u,v) := |L_\tau(u) \cap L_\tau(v)|$. We write $v \sim_{\tau, c} u$ for vertices $v, u$ when  $v \sim u$  and $c \in L_\tau(v) \cap L_\tau(u) $. Furthermore, for  $U \subseteq V \setminus V_\tau$,   let $\tau|_U \coloneqq \{ vc \in \tau: v \in U\}$.

\subsection{Diagonal Matrix Bounds}
To demonstrate the essence of our approach better, we start by restricting our attention to when the matrix bounds $\{M_\tau\}_{\tau \in X: \codimfunc (\tau) \geq 2}$ in  \cref{thm:fulltechnical} are diagonal matrices.  Using diagonal matrix bounds, we analyze the Glauber dynamics for $(1+\epsilon)\Delta$-extra-color  vertex-list-coloring instances. 
\begin{theorem}\label{thm:vertex-coloring-2delta}
Suppose $(G, L)$ is a  $(1+\eps)\Delta$-extra-color    vertex-list-coloring instance for an $0<\eps \leq 1$ such that  $\frac{\ln (\Delta) + 2}{\Delta } \leq \frac{\eps^2}{40}$, and let  $(X, \pi_{n-1})$ be its associated  weighted simplicial complex. For any $2 \leq k\leq n$ and   $\tau \in X$ of co-dimension $k$ we have
$ \lambda_2(P_\tau) \leq \frac{ \frac{5}{2\eps}}{k-1}$.
\end{theorem}
Combined with local-to-global theorems (\cref{thm:localtoglobal,thm:boundeddegspinlocaltoglobal}), this yields a mixing time of $O(n\log n)$ for bounded-degree graphs, and $n^{O(1/\epsilon)}$ in general, in this setting where we have at least $(1+\epsilon)\Delta$ additional colors available to each vertex. Again, we emphasize that this mixing result in itself is not new; a simple coupling argument can already recover $O(n\log n)$ mixing for $(\Delta+1)$-extra-color vertex-list-coloring instances. However, we will see later on how our proof technique can be used to obtain new mixing results for sampling edge colorings which, to the best of our knowledge, cannot be recovered via simple coupling arguments.

To prove the above statement, first for any $\tau$ of co-dimension $2$, we find a diagonal matrix  $F_\tau$ such that $\Pi_\tau P_\tau \preceq 2 \pi_\tau \pi_\tau^\top + \Pi_\tau F_\tau$.
Then, for all $\tau$ of co-dimension at least $3$,  we  apply \cref{thm:fulltechnical} to $M_{\tau} = \frac{\Pi_{\tau}F_{\tau}}{k-1}$ to find a sufficient condition on the diagonal matrix $F_\tau$ based on matrices $F_{\tau'}$ for faces  $\tau \subsetneq \tau'$, to get  $\lambda_2(P_\tau) \leq \frac{\rho(F_\tau)}{k-1}$.
To  find $F_\tau$ for faces  $\tau$ of co-dimension $2$, we  state a more general proposition  that is also useful  for approaches that use  non-diagonal matrix bounds. 
\begin{proposition}\label{prop:basecase_main}
Given a vertex-list-coloring instance $(G, L)$, consider the  weighted complex $(X, \pi_{n-1})$.  For any face $\tau$  of co-dimension 2 such that $G_\tau=(\{u,v\},\{uv\})$ is connected we have, 
\begin{align}\label{eq:basecasegoal1}
 \Pi_\tau P_\tau - 2 \pi_\tau \pi_\tau^\top \preceq  
 \sqrt{\Pi_\tau} \widetilde{M_\tau} \sqrt{\Pi}_\tau,
\end{align}
where $\widetilde{M}_\tau$ is a block diagonal matrix  with a block $\widetilde{M}_{\tau}^{c}$ for every color $c$ such that
$$
\widetilde{M}^c_\tau = \begin{pmatrix}
\frac{1}{(l_\tau(u)-1)(l_\tau(v)-1)} & \frac{-1}{ \sqrt { (l_\tau(u) -1) (l_\tau(v) -1)}}\\
\frac{-1}{ \sqrt { (l_\tau(u) -1) (l_\tau(v) -1)}} & \frac{1}{(l_\tau(u)-1)(l_\tau(v)-1)}
\end{pmatrix}
$$
and all other entries are $0$. 
\end{proposition}
\begin{proof}
 For clarity, we drop $\tau$ from all notation in the proof.
Write $\widetilde{M} =  \widetilde{M}_d + \widetilde{M}_o$, where $\widetilde{M}_d=\diag(\widetilde{M})$ and $\widetilde{M}_o=\offdiag(\widetilde{M})$.
 First observe that for $c\in L(u)$,
 $$ \pi(uc) = \begin{cases}\frac{l(v)-1}{2(l(u)l(v)-l(u,v))} & \text{if }c\in L(u)\cap L(v) \\ \frac{l(v)}{2(l(u)l(v)-l(u,v))}&\text{otherwise}\end{cases} $$
 and a similar identity  holds for any $c\in L(v)$.
 Also, observe that $$\Pi P = \frac{J-J^u-J^v}{2(l(u)l(v)-l(u,v))}+\sqrt{\Pi}\widetilde{M}_o\sqrt{\Pi},$$ 
 where $J,J^u,J^v$ are the all-ones matrix, all-ones matrix on $uc$ rows/columns and all-ones matrix on $vc$ rows/columns, respectively. So, subtracting $\widetilde{M}_o$ from  both sides of \eqref{eq:basecasegoal1} and multiplying by $2(l(u)l(v)-l(u,v))$ it is enough to show
 \begin{align}\label{eq:basecasegoal2}
     J-J^u-J^v -4(l(u)l(v)-l(u,v))\pi\pi^\top \preceq 2(l(u)l(v)-l(u,v))\sqrt{\Pi}\widetilde{M}_d\sqrt{\Pi} =: N_d 
 \end{align}
Write $\ell=l(v)\bone^u+l(u)\bone^v$. 
Also, let $s\in \R^{l(u)+l(v)}$ where $s(xc)=1$ if $c\in L(u)\cap L(v)$ and $s(xc)=0$ otherwise for $x \in \{u,v\}$. Then, by \cref{crosstermfact} we can write,
 \begin{align*} 4(l(u)l(v)-l(u,v))\pi\pi^\top &= \frac{(\ell-s)(\ell-s)^\top}{l(u)l(v)-l(u,v)} \underset{\cref{crosstermfact}}{\underbrace{\succeq}} \frac{\ell\ell^\top+ss^\top-\frac12\ell\ell^\top-2ss^\top}{l(u)l(v)-l(u,v)} 
\\& \succeq \frac{\ell\ell^\top}{2l(u)l(v)}-\frac{ss^\top}{l(u)l(v)-l(u,v)}
 \end{align*} 
 Plugging this into \eqref{eq:basecasegoal2} it is enough to show that
 \begin{align}\label{eq:basecasegoal3}
     J-J^u-J^v+\frac{ss^\top}{l(u)l(v)-l(u,v)}=\bone^u{\bone^v}^\top + \bone^v{\bone^u}^\top+\frac{ss^\top}{l(u)l(v)-l(u,v)} \preceq \frac{\ell\ell^\top}{2l(u)l(v)} +N_d
 \end{align}
 First, observe that by another application of \cref{crosstermfact}, $l^2(v)\bone^u{\bone^u}^\top + l^2(u)\bone^v{\bone^v}^\top\succeq l(u)l(v)(\bone^u{\bone^v}^\top+\bone^v{\bone^u}^\top)$. So,
 $$
 \frac{\ell\ell^\top}{2l(u)l(v)}=\frac{(l(v)\bone^u+l(u)\bone^v)(l(v)\bone^u+l(u)\bone^v)^\top}{2l(u)l(v)}\succeq \bone^u{\bone^v}^\top + \bone^v{\bone^u}^\top
 $$
 Let $I^{\cap}\in\R^{(l(u)+l(v))\times (l(u)+l(v))}$ be the identity matrix only on entries $xc,xc$ where $x\in\{u,v\}$ and $c\in L(u)\cap L(v)$. 
Finally, \eqref{eq:basecasegoal3} simply follows from the fact that 
$$ \frac{ss^\top}{l(u)l(v)-l(u,v)}\preceq \frac{l(u,v)}{l(u)l(v)-l(u,v)}I^{\cap} \preceq N_d$$ 
 where the first inequality uses that the only non-zero rows of $ss^\top$ correspond to a common color and the sum of the entries of any such row is exactly $l(u,v)$ and the last inequality uses that $\frac{l(u,v)}{l(u)l(v)-l(u,v)}\leq \frac{1}{\max\{l(u),l(v)\}-1}$ and that $N_d(uc,uc)=\frac{1}{l(u)-1},N_d(vc,vc)=\frac{1}{l(v)-1}$ if $c\in L(u)\cap L(v)$ and it is zero otherwise.
\end{proof}

Note that in the above proposition, $ \widetilde{M}_\tau  \preceq (\frac{1}{\beta} + \frac{1}{\beta^2})I^{X_\tau(0)}$, which gives us the diagonal matrix $F_\tau$ for any $\tau$ of co-dimension 2 such that $G_\tau$ is connected.  
Now, using \cref{thm:fulltechnical}, we derive  a set of sufficient conditions on the family $\{F_\tau\}_{\tau \in X: \codimfunc(\tau) \geq 2}$   to get $\lambda_2(P_\tau) \leq \frac{\rho(F_\tau)}{k-1}$ for all $\tau$ of co-dimension $2 \leq k \leq n$. 
 \begin{proposition}\label{prop:Fcondition-2delta}
 Given a $\beta$-extra-color  vertex-list-coloring instance $(G, L)$, with corresponding  weighted simplicial complex $(X, \pi_{n-1})$, suppose  $\{F_\tau \in \mathbb{R}^{X(0) \times X(0)}\}_{\tau \in X: \codimfunc(\tau)\geq 2}$  is a  family of diagonal matrices supported on $X_\tau (0) \times X_\tau (0)$ such that  $F_\tau=  f_{\times} (X_\tau, \{F_{ \tau \cup \sigma} \}_{ \emptyset \subsetneq \sigma \in X_\tau(\leq \codimfunc(\tau)-3)})$ if $G_\tau$ is disconnected and otherwise,
 \begin{enumerate}
     \item For all $\tau$ of co-dimension $2$:   $F_\tau (vc, vc) = \frac{1}{\beta} +\frac{1}{\beta^2}$ for all  $vc \in X_\tau (0)$.
     \item For all $\tau$  of co-dimension  $k \geq 3$:  $F_{\tau} \preceq  \frac{(k-1)^2}{3k-1}I^{X_\tau(0)}$ and for all $vc \in X_\tau (0)$ 
             \begin{align*}\sum_{uc' \in X_{\tau \cup vc}(0)} p(uc'|\tau \cup vc ) F_{\tau\cup uc'} (vc, vc) \leq (k-2) F_\tau(vc,vc) - F_\tau^2(vc,vc). 
             \end{align*}
  \end{enumerate}

 Then, for all $k\geq 2$ and $\tau$ of co-dimension $k$, $\lambda_2(P_\tau) \leq \frac{\rho(F_\tau)}{k-1}$. 
 \end{proposition}
 
\begin{proof}
We prove that  the conditions of   \cref{thm:fulltechnical} hold  for $M_\tau \coloneqq  \frac{\Pi_\tau F_\tau}{k-1}$ for any face $\tau$ of co-dimension at least 2.
The desired condition  holds for  any $\tau$ of co-dimension $2$ by \cref{prop:basecase_main}.
Now, let $k \geq 3$.  First assume that $G_\tau$ is disconnected with connected components  $G_\tau[U_1], \dots, G_\tau[U_\ell]$ and   
associated  complexes  $Y_{1},\dots Y_{\ell}$. We can write $(X, \pi_{\tau, k-1}) =(Y_{1}, \mu_1) \times\dots\times (Y_{\ell}, \mu_\ell)$, where $\mu_{i}$  is the uniform distribution over facets of $Y_i$. For an $\alpha \in X_\tau(k -1)$ 
let  $\alpha_{-i} \coloneqq \alpha \setminus \alpha|_{U_i}$. 
Therefore, 
\begin{align*}
    \sum_{1 \leq i \leq  \ell: d_{Y_i} \geq 1} \frac{d_{Y_i}  (d_{Y_i}+1)}{(\edgecodim-1) \edgecodim} M_{\tau  \cup \alpha_{-i}} &\underset{\text{def of } M_{\tau\cup_{\alpha_{-i}}}} = \sum_{1 \leq i \leq  \ell: d_{Y_i} \geq 1} \frac{d_{Y_i}  (d_{Y_i}+1)}{(\edgecodim-1) \edgecodim} \frac{\Pi_{\tau\cup_{\alpha_{-i}}}}{d_{Y_i}} F_{\tau  \cup \alpha_{-i}}\\
    &=\sum_{1 \leq i \leq  \ell: d_{Y_i} \geq 1} (\Pi_\tau)^{X_{\tau \cup \alpha_{-i}}(0)} \frac{F_{\tau\cup{\alpha_{-i}}}}{k-1}\underset{\text{def of }F_\tau}{=} \frac{\Pi_\tau F_\tau}{k-1}=M_\tau
\end{align*}
as desired.

Now, assume that $G_\tau$ is connected. 
Note that since each entry of $F_\tau$ is at most $\frac{(k-1)^2}{3k-1}$, we have  $M_\tau \preceq \frac{\edgecodim-1 }{3\edgecodim-1} \Pi_\tau$.
 Therefore, it only remains to show that   $  \mathbb{E}_{vc \sim \pi_\tau }  M_{\tau \cup vc}   \preceq   M_{\tau} - \frac{\edgecodim -1 }{\edgecodim -2}  M_\tau\Pi_\tau^{-1} M_\tau$. This is 
equivalent to showing that
\begin{align*}
    &\Pi_{\tau}^{-1} \mathbb{E}_{uc' \sim \pi_\tau } \left[ \Pi_{\tau \cup uc'}  \frac{\dbound_{\tau \cup uc'}}{\edgecodim - 2 } \right] 
    \preceq   \frac{\dbound_{\tau}}{\edgecodim - 1}
    -  \frac{\dbound^2_{\tau}} {(\edgecodim - 2)(\edgecodim - 1)}. 
    \end{align*}
One can check that   
\begin{align*}
\mathbb{E}_{uc' \sim \pi_\tau } \left[ \Pi_{\tau}^{-1}\Pi_{\tau \cup uc'}  \frac{\dbound_{\tau \cup vc}}{\edgecodim - 2 }\right] (vc, vc) =  \frac{\sum_{uc' \in X_{\tau\cup vc}(0)} p(uc'|\tau \cup vc) F_{\tau \cup uc'}(vc, vc)}{(k-1)(k-2)}. 
\end{align*}
Therefore, it is enough that
\begin{align*}
   \frac{\sum_{uc' \in X_{\tau\cup vc}(0)} p(uc'|\tau \cup vc) F_{\tau \cup uc'}(vc, vc)}{(k-1)(k-2)} \leq \frac{\dbound_{\tau} (vc, vc)}{\edgecodim - 1} -   \frac{\dbound^2_{\tau}(vc, vc) }{(\edgecodim -1)(\edgecodim -2) },
\end{align*}
which holds by assumption.
\end{proof}
Now, to complete the proof of \cref{thm:vertex-coloring-2delta} it only remains to find $\{F_\tau\}_{\tau \in X(\leq n-2)}$ that  satisfies the above conditions.
The proof can be found in the \cref{sec:appendixvercolor}. 
Let us remark  why we need the assumption $\beta>\Delta$ in this proof.
Consider the worst case example, where $G$ is a complete graph  with $\Delta+1$ vertices. In that case, by symmetry, $F_\tau(vc,vc)=\frac1\beta+\frac1{\beta^2}$ for all faces of co-dimension 2, and  every matrix $F_\tau$ is a multiple of identity on $X_\tau(0)\times X_\tau(0)$. So, the conditions on $F_\tau$ reduces to the following systems of inequalities:
$$ (k-1)f(k-2)\leq (k-2)f(k-1)-f(k-1)^2 \quad \forall 3\leq k\leq \Delta, f(1)=\frac1\beta+\frac1{\beta^2}.$$
It is not hard to see that such a system does not have a solution up to  $k=\Delta+1$  when $\beta\leq \Delta$.

\subsection{Vertex-List-Coloring for Trees Using Non-Diagonal Matrix Bounds}
By allowing the matrix bounds $\{M_\tau\}_{\tau \in X: \codimfunc(\tau) \geq 2}$ in  \cref{thm:fulltechnical} to be non-diagonal matrices, one can hope to get a tighter result. In this section, for any constant $\epsilon>0$ we analyze the Glauber dynamics for $\epsilon$-extra-color vertex-List-Coloring instances when the graph is a tree. 
\begin{theorem}\label{thm:vertex-coloring-tree}
Consider  an arbitrary constant $\epsilon >0$ and a $\epsilon\Delta$-extra-color    vertex-list-coloring instance $(G, L)$ such that $G$ is a tree and  $\frac{\ln^2(\Delta)}{ \Delta} \leq \frac{\epsilon^2}{100}$.   For the weighted simplicial complex  $(X, \pi_{n-1})$, any $2 \leq k\leq n$ and   $\tau \in X$ of co-dimension $k$ we have  $ \lambda_2(P_\tau) \leq \frac{\frac{1}{20}+\frac1{\eps}}{k-1}$.
\end{theorem}
For any $k \geq 2$ and  $\tau$ of co-dimension at least $k$, assume that $M_\tau$ is of the form   $$M_\tau = \Pi_\tau \frac{F_\tau}{k-1}+ \sqrt{ \Pi_\tau} \frac{A_\tau}{k-1} \sqrt{ \Pi_\tau}, $$ for a diagonal matrix $F_\tau$ and a hollow matrix  $A_\tau$. 
The goal is to find $F_\tau$ and $A_\tau$ such that  $M_\tau$ satisfies the conditions of \cref{thm:fulltechnical}. This is easily doable if $k=2$ by \cref{prop:basecase_main}.
A natural approach is to 
  define $A_\tau$ for any $\tau$ of co-dimension at least $3$ such that
$$\sqrt{\Pi_\tau} \frac{A_\tau}{k-1} \sqrt{\Pi_\tau}=\mathbb{E}_{vc \sim \pi_\tau} \sqrt{ \Pi_{\tau \cup  vc}} \frac{A_{\tau \cup  vc}}{k-2} \sqrt{ \Pi_{\tau \cup vc}},$$ 
when $G_\tau$ is connected. Note that the following definition is not restricted to trees. 

\begin{definition}[Family of Matrices
$\{A_\tau\}_{\tau \in X: \codimfunc(\tau) \geq 2}$ ] \label{def:A-tree}
Given a vertex-list-coloring instance $(G, L)$ and its associated weighted simplicial complex $(X, \pi_{n-1})$, define  $\{A_\tau\}_{\tau \in X: \codimfunc(\tau) \geq 2}$ as follows: let  $A_\tau \coloneqq f_\times(X, \{A_{\tau \cup \sigma}\}_{\emptyset \subsetneq \sigma \in X_\tau (\codimfunc(\tau)-3)})$ if $G_\tau$ is disconnected and otherwise,

\begin{enumerate}
    \item  For any $\tau$ of co-dimension 2,  say $G_\tau=(\{u,v\},\{uv\})$; define  $A_\tau \in \mathbb{R}^{X(0)  \times X(0) }$ to be a hollow block diagonal matrix with a block for every color such that $A_\tau(uc, vc)= A_\tau(vc, uc)=\frac{-1}{ \sqrt { (l_\tau(u) -1) (l_\tau(v) -1)}},
$ for   $c \in L_\tau(v) \cap L_\tau(u)$, 
and all other entries are $0$. 
\item For any  $\tau$ of co-dimension  $k\geq 3$, let
\begin{align}\label{eq:A_tree}
    A_\tau \coloneqq \frac{k-1}{k-2} \sqrt{\Pi_\tau^{-1}} \left( \mathbb{E}_{vc \sim \pi_\tau} \sqrt{\Pi_{\tau \cup vc }} A_{\tau \cup vc} \sqrt{\Pi_{\tau \cup vc }} \right ) \sqrt{\Pi_\tau^{-1} }. 
\end{align} 
\end{enumerate}

Observe that $A_{\tau}$ is symmetric and hollow. Furthermore,  its non-zero entries correspond to $v \sim_{\tau, c} u$, and  when $G_\tau$ is connected,
\begin{align*}
A_\tau (vc, uc)   = \frac{1}{k-2} \sum_{wc' \in X_\tau(0) : w \neq v, u} \sqrt{p( wc'| vc \cup \tau)  p(wc'| uc \cup \tau) } A_{\tau \cup wc'} (vc, uc).
\end{align*}
\end{definition}

We can bound entries of $A_\tau$ for any $\tau \in X(\leq n-3)$ as follows. 

\begin{proposition}\label{prop:Abound_tree}
Consider a $\beta$-extra-color   vertex-list-coloring instance $(G, L)$,
For any face $\tau$ of co-dimension at least 2 and $uc, vc \in X_\tau (0)$ such that $v \sim_c u$, 
$-\frac{1}{\beta} \leq  A_{\tau}(vc, uc) \leq  0$. 
\end{proposition}
\begin{proof}
Let $\tau$ be a face of co-dimension $k \geq 2$. We prove by induction on $k$. It clearly holds for $k=2$ by definition. For $k >2$ and $vc, uc \in X_\tau (0)$, we have
\begin{align*}
A_\tau (vc, uc)   &= \frac{1}{k-2} \sum_{ w\in V_{\tau}: w \neq u, v} \sum_{c' \in L_\tau(w)} \sqrt{p( wc'| vc \cup \tau)  p(wc'| uc \cup \tau) } A_{\tau \cup wc'} (vc, uc) \\&\geq  \frac{1}{k-2} \sum_{ w\in V_{\tau}: w \neq u, v} \sum_{c' \in L_\tau(w)} \sqrt{p( wc'| vc \cup \tau)  p(wc'| uc \cup \tau) } \cdot  \frac{-1}{\beta } \tag{by IH}\\& \geq
 \frac{1}{k-2}\cdot\frac{-1}{\beta} \sum_{w\in V_{\tau}: w \neq u, v} \left(\sum_{c' \in L_\tau(w)} p( wc'| vc \cup \tau)\right) \left(\sum_{c' \in L_\tau(w)} p( wc'| uc \cup \tau)\right) \tag{by Cauchy-Schwarz } \\&\geq \frac{1}{k-2}\cdot  \frac{-1}{\beta } \sum_{w\in V_{\tau}: w \neq u, v} 1  = -\frac{1}{\beta}. 
\end{align*}
One can further see that  $A_\tau (vc, uc) \leq 0$ follows from  the induction hypothesis. 
\end{proof}

Now, we apply \cref{thm:fulltechnical} to  derive  sufficient conditions on the family $\{F_\tau\}_{\tau \in X, \text{codim}(\tau)\geq 2}$  to get $\lambda_2(P_\tau) \leq \frac{\rho(F_\tau + A_\tau)}{k-1}$ for all $\tau$ of co-dimension $2 \leq k \leq n$. 

 \begin{proposition}\label{prop:Fcondition-trees}
 Let $(G, L)$   be a $\beta$-extra-color  vertex-list-coloring instance such that $G$ is a tree, and let $(X, \pi_{n-1})$ be its associated weighted complex. Take an arbitrary vertex and make $G$  a rooted tree with root vertex $r$. For $v \neq r$, we write $a(v)$ to denote immediate ancestor of $v$, i.e., parent of $v$.  Let $\{F_\tau \in \mathbb{R}^{X(0) \times X(0)}\}_{\tau \in X: \codimfunc(\tau) \geq 2}$ be a family of diagonal matrices supported on $X_\tau (0) \times X_\tau (0)$ such that  $F_\tau=  f_{\times} (X_\tau, \{F_{ \tau \cup \sigma} \}_{ \emptyset \subsetneq \sigma \in X_\tau(\leq \codimfunc(\tau)-3)})$ if $G_\tau$ is disconnected and otherwise,
 \begin{enumerate}
     \item For all $\tau$ of co-dimension $2$:  $F_\tau$ is defined as   $F_\tau (vc, vc) =  \frac{1}{\beta^2}$ for  $vc \in X_\tau (0)$.
     \item For any $\tau$ of co-dimension $k\geq 3$: $F_\tau \preceq (\frac{(k-1)^2}{3k-1} - \frac{1}{\beta})I^{X_\tau(0)}$, and for all $vc\in X_\tau(0)$,
     \begin{align}\label{eq:F_tree_assump}
        \sum_{uc' \in X_{\tau\cup vc}(0)} p(uc'|\tau, vc) F_{\tau\cup uc'} (vc, vc) \leq (k-2) F_\tau(vc, vc) - 2F_\tau^2 (vc, vc)- \gamma_\tau(vc),
        \end{align}
        where $\gamma_\tau( vc) =  \frac{4 \Delta_\tau (v)}{\beta^2}$ if  $v$ is the root of the rooted tree $G_\tau$, and  $\gamma_\tau  ( vc) =  \frac{4 \left(\Delta_\tau (v) +\Delta_\tau (a(v) )-1\right) }{\beta^2}$ otherwise.
 \end{enumerate}
 Then, for all $k\geq 2$ and $\tau$ of co-dimension $k$, $\lambda_2(P_\tau) \leq \frac{\rho(F_\tau + A_\tau)}{k-1} $, where $A_\tau$  is defined in \cref{def:A-tree}.
 \end{proposition}
 
 \begin{proof}
We prove that  the conditions of   \cref{thm:fulltechnical} hold  for $M_\tau \coloneqq  \Pi_\tau\frac{F_\tau}{k-1} + \sqrt{\Pi_\tau}\frac{A_\tau}{k-1}\sqrt{\Pi_\tau}$ for $\tau \in X( \leq n-3)$.
Note that  the desired condition  holds for  any $\tau$ of co-dimension $2$ by definition.
Now, take $k \geq  2$.  Assume $G_\tau$ is disconnected. Using the definition of  $A_\tau$ and our assumption about $F_\tau$, the proof of this case  is   similar  to what we argued in  \cref{prop:Fcondition-2delta}.  
Now, assume that $G_\tau$ is connected.
  Note   that by \cref{prop:Abound_tree},  the absolute value of every off-diagonal entry of $A_\tau$ is at most $\frac{1}{\beta}$ and that there are at most $(k-1)$ non-zero entries per row. Therefore, $\sqrt{\Pi_\tau} A_\tau \sqrt{\Pi_\tau} \preceq \frac{1}{\beta}\Pi_\tau$. Since each entry of $F_\tau$ is at most $\frac{(k-1)^2}{3k-1} - \frac{1}{\beta}$, we have  $M_\tau \preceq \frac{\edgecodim-1 }{3\edgecodim-1} \Pi_\tau$.
 Therefore, it only remains to show that   $  \mathbb{E}_{vc \sim \pi_\tau }  M_{\tau \cup vc}   \preceq   M_{\tau} - \frac{\edgecodim -1 }{\edgecodim -2}  M_\tau\Pi_\tau^{-1} M_\tau$. This is 
equivalent to showing that
\begin{align}\label{eq:desired_tree}
    &\sqrt{\Pi_{\tau}^{-1}} \mathbb{E}_{vc \sim \pi_\tau } \left[ \Pi_{\tau \cup vc}  \frac{\dbound_{\tau \cup vc}}{\edgecodim - 2 } + \sqrt{\Pi_{\tau \cup vc}} \frac{A_{\tau \cup vc}}{\edgecodim-2} \sqrt{\Pi_{\tau \cup vc}}\right]  \sqrt{\Pi_{\tau}^{-1}}
    \preceq  & \frac{\dbound_{\tau}}{\edgecodim - 1} + \frac{A_{\tau}}{\edgecodim-1} 
    -  \frac{(\dbound_{\tau} + A_{\tau})^2 } {(\edgecodim - 2)(\edgecodim - 1)}. 
\end{align}
We starting by proving an upper bound on $A_\tau^2$. 
Define $A_{\tau, even}(vc, uc)= A_{\tau, even}(uc, vc) \coloneqq A(vc, uc)$ for all $vc, uc \in X_\tau (0)$ such that $v$ is at an even distance from the root of $G_\tau$ and $a(u) = v$, and let other entries be $0$. Define $A_{\tau, odd} \coloneqq A_{\tau, even} - A_{\tau, odd}$. By   \cref{crosstermfact},
\begin{align*}
 A_\tau^2 \preceq 2 A^2_{\tau, even} + 2 A_{\tau, odd}^2.
\end{align*}
Furthermore, $A^2_{\tau, even} (vc, uc)  = A^2_{\tau, even} (uc, vc) \neq 0$ only if $v = u$ and $v$ is at an even distance from the root or when $a(v) = a(u)$ and $u, v$ are at an odd distance from the root. A similar fact holds for $A^2_{\tau, odd}$. 
Therefore, if we let $\gamma_\tau(wc) =0$ for $wc \notin X_\tau (0)$, we get
\begin{align*}
   4 ( A^2_{\tau, even}  + A^2_{\tau, odd}) \preceq \diag (\gamma_\tau),
\end{align*}
where we applied \cref{fact:rowsum}, and used  \cref{prop:Abound_tree} to bound the absolute value of the entries of $A^2_{\tau, even}$  and  $A^2_{\tau, odd}$. Therefore, by \cref{crosstermfact}, the RHS of \cref{eq:desired_tree} is bounded by
 \begin{align}\label{eq:rhsbound_tree}
     \frac{\dbound_{\tau}}{\edgecodim - 1} + \frac{A_{\tau}}{\edgecodim-1} 
    -  \frac{2\dbound_{\tau}^2 + 2A_{\tau}^2 } {(\edgecodim - 2)(\edgecodim - 1)} \succeq    \frac{\dbound_{\tau}}{\edgecodim - 1} + \frac{A_{\tau}}{\edgecodim-1} 
    - \frac{2 F^2 _\tau +\diag (\gamma_\tau) }{(k-1)(k-2)}.
 \end{align}
On the other hand, by the definiton of $A_\tau$ (see \cref{eq:A_tree}), the LHS of \cref{eq:desired_tree} is equal to 
\begin{align*}
\mathbb{E}_{uc' \sim \pi_\tau } \left[ \Pi_{\tau}^{-1}\Pi_{\tau \cup uc'}  \frac{\dbound_{\tau \cup uc'}}{\edgecodim - 2 } \right]+ \frac{A_\tau}{k-1}.
\end{align*}
Furthermore, for $vc \in X_\tau (0)$, 
\begin{align*}
\mathbb{E}_{uc' \sim \pi_\tau } \left[ \Pi_{\tau}^{-1}\Pi_{\tau \cup uc'}  \frac{\dbound_{\tau \cup uc'}}{\edgecodim - 2 }\right](vc, vc) = \frac{ \sum_{uc' \in X_{\tau\cup vc}(0)} p(uc'|\tau \cup vc) F_{\tau\cup uc'} (vc, vc)}{(k-1)(k-2)}
\end{align*}
Combining this with \cref{eq:rhsbound_tree}, the desired inequality in \cref{eq:desired_tree} follows from the assumption (see \cref{eq:F_tree_assump}). 
\end{proof}
Finally, with this in hand, we prove \cref{thm:vertex-coloring-tree} similar to what we did for \cref{thm:vertex-coloring-2delta}. The proof can be found in \cref{sec:appendixvercolor}.

\section{Edge Coloring}

Consider a graph  $G = (V, E)$ and a function $L: E \rightarrow 2^{[q]}$. The pair $(G, L)$ is called an edge-list-coloring instance. 
For a vertex $v$ and an edge $e$, we write $e \sim_c v$ when $e \sim v$ and  $c \in L(e)$.
Furthermore, for any $e,  f \in E$, we write  $e \sim_c f$ when $e \sim f$  and $c \in L(e) \cap L(f)$. 
Furthermore, we define a $\beta$-extra-color edge-list-coloring instance as follows. 
\begin{definition}
We say an edge-list-coloring instance $(G, L)$ is a $\beta$-extra-color instance if  for each $e  \in E$, $|L(e)| \geq \beta + \Delta_G(e)$. 
\end{definition}

An assignment $\sigma: E \rightarrow [q]$ is a $L$-edge-list-coloring of $G$ if $\sigma (e) \in L(e)$ for all $e \in E$. 
We say  $\sigma$ is proper if $\sigma (e) \neq \sigma(f)$ whenever $e \sim f$. When it is clear from context we say $\sigma$ is a proper coloring to mean it is a proper $L$-edge-list-coloring.  We say $\tau$ is   proper partial coloring on $H \subset E$ when it is a proper $L|_H$-edge-list-coloring for $(V, H)$.  We may view a proper coloring as a set of edge-color pairs $(e, c)$ which  we denote by $ec$ for simplicity of notation. We denote  the uniform distribution over proper $L$-edge-list-colorings of $G$ by $\pi_{m-1}$ when $(G, L)$ are clear from context. For a proper partial coloring on $H \subset E$ and $e \in E \setminus H$, define    
\begin{align*}
    p(ec | \tau) \coloneqq \mathbb{P}_{\sigma \sim \pi_{m-1}} (\sigma(e)=c |\forall f \in H: \sigma (f) = \tau(f) ). 
\end{align*}

\par
To analyze the Glauber dynamics on an edge-list-coloring instance we  associate a simplicial complex to it. 

\begin{definition}[Simplicial Complex of an Edge-List-Coloring Instance]
 Given an edge-list-coloring instance $(G, L)$, let $X(G, L)$ be a pure $(m -1)$-dimensional simplicial complex specified by the following facets:  $\{(e, \sigma(e))\}_{e \in E} $ is a facet  if and only if $\sigma$  is a proper $L$-edge-list-coloring for $G$. 
\end{definition}

When it is clear from context, we abbreviate $X (G, L)$ to $X$. Note that for all $0 \leq k \leq m$, any  face  $\tau$ of co-dimension $k$  is a partial coloring on a subset of edges $H$ of size $m-k$ ($k$ edges remain uncolored).   Furthermore,  $X_\tau (0)$ can be seen as the set of all $ec$ such that $c \in L(e)$, $e \notin H$ and for any $f \sim e$, $fc \notin \tau$. 
Analogous to vertex-list-colorings, the   Glauber dynamics  on  $(G, L)$ is the down-up walk on the facets of $(X, \pi_{m-1})$. So, as we did before for vertex-list-colorings, our aim is to  apply  \cref{thm:fulltechnical} to the simplicial complex to bound the second eigenvalue of the transition probability matrix of the local walks and then apply \cref{thm:localtoglobal} to get a bound for the transition probability matrix of the down-up walk on the facets. The following  is the main theorem of this section. 

\begin{theorem}\label{thm:edge-coloring}
Let $(G, L)$  be a   $(\frac{4}{3} + 4 \eps) \Delta$-extra-color edge-list-coloring instance 
for some $0<\epsilon \leq \frac{1}{10}$
such that  $\frac{\ln^2(\Delta)}{ \Delta} \leq \frac{\eps^3}{15}$, and let $(X, \pi_{m-1})$ be its associated   weighted simplicial complex. For any $2 \leq k\leq m$ and   $\tau \in X$ of co-dimension $k$ we have
$ \lambda_2(P_\tau) \leq \frac{\eps+\frac1{\eps}}{k-1}$.
\end{theorem}
We remark that our analysis here is not tight and we expect that the factor $4/3$ can be improved with a more careful analysis. 

We proceed by introducing  some notation and definitions. 
Given a  face $\tau \in X$, let $E_\tau$ be the set of uncolored edges, i.e. $E_\tau \coloneqq \{ e : \exists c, ec\in X_\tau(0)\}$. Let $G_\tau = (V, E_\tau)$  and $\Delta_\tau(.)$ be the degree function of $G_\tau$. Similarly, if $e =\{u, v\}$,  define $\Delta_\tau (e)$ to be number of edges in $G_\tau$ that share an endpoint with $e$, i.e.  $\Delta_\tau (e)= \Delta_\tau (u) +\Delta_\tau (v) -2$. We define $L_\tau (e) \coloneqq \{ c \in L(e) : ec \in X_\tau (0) \}$. Let   $ l_\tau (e) \coloneqq |L_\tau(e)|$ and  $l_\tau(e,f) := |L_\tau(e) \cap L_\tau(f)|$. Furthermore, we write $e \sim_{\tau, c} v$ when $e\sim v$ and $c \in L_\tau(e)$. Similarly, define $e \sim_{\tau, c} f$ for edges $e$ and $f$.   
Finally, for  any matrix $B \in \mathbb{R}^{ X(0) \times X(0)}$, define  the restriction of $B$ to $v  \in V$ as $B^v(ec, fc) \coloneqq B(ec, fc)$ for any $e ,f \sim v$, and  $0$ on all other entries. 
  Let $B^c \in  \mathbb{R}^{ X(0) \times X(0)}$ be  defined as $B^c(ec, fc) \coloneqq B(ec, fc) $  for all $e \sim_{ c} f$, and  $0$ on all other entries.

\par
Now, similar to our approach to vertex-coloring for trees, 
for any $k \geq 2$ and face  $\tau$ of co-dimension $k$, assume that $M_\tau$ is of the form  
\begin{align}\label{eq:bound_edge}
    M_\tau = \Pi_\tau \frac{F_\tau}{k-1}+ \sqrt{ \Pi_\tau} \frac{A_\tau}{k-1} \sqrt{ \Pi_\tau},
\end{align} 
for a diagonal matrix $F_\tau$ and a hollow matrix  $A_\tau$. 
The goal is again to find $F_\tau$ and $A_\tau$ such that  $M_\tau$ satisfies the conditions of \cref{thm:fulltechnical}. For $k=2$,  \cref{prop:basecase_main} gives us such matrices. For  $k \geq 3$, as opposed to what we did for vertex-coloring of trees, we let $\sqrt{\Pi_\tau} \frac{A_\tau}{k-1} \sqrt{\Pi_\tau}$  deviate from $\mathbb{E}_{vc \sim \pi_\tau} \sqrt{ \Pi_{\tau \cup  vc}} \frac{A_{\tau \cup  vc}}{k-2} \sqrt{ \Pi_{\tau \cup vc}}$ in order to control the growth of $F_\tau$.

\begin{definition}[Family of Matrices $\{A_{\tau,\eps}\}_{\tau \in X, \text{codim}(\tau)\geq 2}$ ] \label{defA}
Let $(G, L)$ be a $\beta$-extra-color edge-list-coloring instance, and  let$(X, \pi_{m-1})$ be its associated  weighted   complex. 
For  $\eps>0$, define  $\{A_{\tau,\eps}\}_{\tau \in X,\codimfunc(\tau)\geq2}$ as follows:  let  $A_{\tau, \eps} \coloneqq f_\times(X, \{A_{\tau \cup \sigma, \eps}\}_{\emptyset \subsetneq \sigma \in X_\tau (\codimfunc(\tau)-3)})$ if the line graph of  $G_\tau$ is disconnected and  otherwise, 
\begin{enumerate}
    \item  For any face $\tau$ of co-dimension $2$,  
let  $A_{\tau, \eps} \in \mathbb{R}^{X(0)  \times X(0) }$ be a hollow block diagonal matrix with a block for every color such that 
$$
A_{\tau,\eps}(ec, fc)  = A_{\tau,\eps}(fc, ec) \coloneqq -\frac{1}{ \sqrt { (l_\tau(e) -1) (l_\tau(f) -1)} },
$$
for $e, f \in E_\tau$ and  any $c \in L_\tau(e) \cap L_\tau(f)$, 
and all other entries are $0$. 
\item For any $k \geq 3$ and a face $\tau$ of co-dimension  $k$, define
    $A_{\tau,\eps}  \coloneqq \linear_{\tau, \eps} + \frac{\offdiag(\sqr_{\tau,\eps})}{\edgecodim-2},$
where $\linear_{\tau, \eps}$ and $\sqr_{\tau,\eps}$ are defined as follows:
\begin{align}\label{eq:Ltaudefedgecolor}
    \linear_{\tau, \eps} &\coloneqq \frac{\edgecodim-1}{\edgecodim-2} \sqrt{\Pi_\tau^{-1}} \left( \mathbb{E}_{gc \sim \pi_\tau} \sqrt{\Pi_{\tau \cup gc }} A_{\tau \cup gc, \epsilon} \sqrt{\Pi_{\tau \cup gc }} \right ) \sqrt{\Pi_\tau^{-1} },
\end{align}

\begin{align}
\label{eq:defedgecoloringStau}
 \sqr^v_{\tau,\eps} &\coloneqq
\begin{cases}
        4(1+\epsilon)  \left( (\linear_{\tau, \eps}^{+, v})^2 +(\linear_{\tau, \eps}^{-, v})^2  \right) & \text{if } \degree_\tau(v) \leq \frac{\beta}{4(1+\epsilon)}  , \\
        2 (1+\epsilon)   (\linear_{\tau, \eps}^{v})^2  & otherwise.
    \end{cases}
\end{align}
and $S_{\tau,\eps}=\sum_{v} S^v_{\tau,\eps}$.
\end{enumerate}
 Observe that all three matrices $\linear_{\tau,\eps},S_{\tau,\eps},A_{\tau,\eps}$ are symmetric and hollow. Furthermore, the non-zero entries of these matrices correspond to $e,f\sim_{\tau,c} v$, for some $v\in G$ and when the line graph of $G_\tau$  is connected,
\begin{align}\label{eq:defedgecoloringLtau}
\linear_{\tau, \eps}(ec, fc)   = \frac{1}{\edgecodim-2} \sum_{gc' \in X_\tau(0) : g \neq e, f} \sqrt{p( gc'| ec \cup \tau)  p(gc'| fc \cup \tau) } A_{\tau \cup gc', \eps} (ec, fc).
\end{align}
When it is clear from context, we drop $\epsilon$ from the subscripts of matrices defined above . 
\end{definition}

\begin{lemma}\label{rec_closefrom}
For any $\tau \in X$, and $ec, fc \in X_\tau(0)$, if the line graph of $G_\tau$ is connected, then 

\begin{align*}
\underset{g\in E_\tau,g\neq e,f}{\textup{avg}}\min_{c'\in L_\tau(g)} A_{\tau\cup gc', \eps}(ec,fc)\leq \linear_{\tau, \eps}(ec,fc)\leq \underset{g\in E_\tau,g\neq e,f}{\textup{avg}}\max_{c'\in L_\tau(g)} A_{\tau\cup gc', \eps}(ec,fc)
\end{align*}
\end{lemma}
\begin{proof}
Let $\tau$ be a face of co-dimension $k \geq 2$.  It clearly holds for $k=2$ by definition. For $k >2$ and $ec, fc \in X_\tau (0)$, we have
\begin{align*}
\linear_\tau (ec, fc)   &= \frac{1}{k-2} \sum_{ g\in E_{\tau}: g \neq e, f} \sum_{c' \in L_\tau(g)} \sqrt{p( gc'| ec \cup \tau)  p(gc'| fc \cup \tau) } A_{\tau \cup gc'} (ec, fc) \\&\leq  \frac{1}{k-2} \sum_{ g\in E_{\tau}: g \neq e, f} \left(\sum_{c' \in L_\tau(g)} \sqrt{p( gc'| ec \cup \tau)  p(gc'| fc \cup \tau) } \right) \cdot \left(\max_{c'\in L_\tau(g)} A_{\tau\cup gc'}(ec,fc)\right)\\& \leq
 \frac{1}{k-2} \sum_{g\in E_{\tau}: g \neq e,f} \left( \left(\sum_{c' \in L_\tau(e)} p( gc'| ec \cup \tau) \right) \left(\sum_{c' \in L_\tau(g)} p( gc'| fc \cup \tau)\right)\right) \left(\max_{c'\in L_\tau(g)} A_{\tau\cup gc'}(ec,fc)\right) \tag{by Cauchy-Schwarz } \\&\leq \frac{1}{k-2}\sum_{g\in E_{\tau}: g \neq e,f} \max_{c'\in L_\tau(g)} A_{\tau\cup gc'}(ec,fc)  = \underset{g\in E_\tau,g\neq e,f}{\textup{avg}}\max_{c'\in L_\tau(g)} A_{\tau\cup gc'}(ec,fc). 
\end{align*}
The other side of the inequality follows from a similar argument. 
\end{proof}

In order to find diagonal matrices $\{F_\tau\}_{\tau \in X: \text{codim}(\tau)\geq 2}$ such that $\{M_\tau\}_{\tau \in X: \text{codim}(\tau)\geq 2}$ as defined by \cref{eq:bound_edge}   satisfies the conditions of \cref{thm:fulltechnical}, we would  need to prove some bounds on  the entries of $\{A_{\tau,\eps}\}_{\tau \in X, \text{codim}(\tau)\geq 2}$ and $\{S_{\tau,\eps}\}_{\tau \in X, \text{codim}(\tau)\geq 2}$.

\begin{proposition}\label{prop:A_bound}
Suppose $(G, L)$ is a  $\beta$-extra-color edge-list-coloring instance  where $\beta = (\frac{4 }{3} + 4\epsilon) \Delta$ for  an $0 <\eps \leq \frac{1}{10}$ such that $ 2\eps^{-2} \leq \Delta $. For any $\tau\in X$ of $\codimfunc(\tau)\geq 2$, 
the matrix $ A_{\tau}$  defined in \cref{defA} satisfies the following: for any  vertex $v\in G$,  color $c$ and $e, f \sim_{\tau, c} v$, 
\begin{enumerate}[label=(\roman*)]
    \item  if $\degree_\tau (v) \leq \frac{\extraColors}{4 (1+\epsilon)}$, then
    $ -\frac{1}{\extraColors} \leq A_{\tau, \eps} (ec, fc) \leq 4 (1+\epsilon) \frac{\degree_\tau(v)-2}{\extraColors^2}$, \label{firstcase}
    \item  otherwise, if $\degree_\tau (v) \geq \frac{\extraColors}{4 (1+\epsilon)}$, then
    $|A_{\tau, \eps} (ec, fc)| \leq \frac{1}{1.5 \extraColors - 2(1+2\epsilon) \degree_\tau(v)}$.\label{secondcase}
\end{enumerate}

\end{proposition}

\def\dctv{\Delta_\tau(v)}
\begin{proof}
Fix a vertex $v$. We prove the claim inductively for any pair of edges incident to $v$. 

{\bf Case \ref{firstcase}.}  
 Let $\tau$ be any face of co-dimension $k\geq 2$. We prove by induction on $\dctv + \edgecodim$. We start with the base case, that is when $\dctv+\edgecodim = 4$, i.e., $\dctv =\edgecodim =2$. It is easy to see that, for any color $c$ and $e,f\sim_{\tau,c} v$, we have 
 $-\frac{1}{\extraColors} \leq A_\tau (ec, fc) \leq 0$, by definition.
Now, we prove the claim for  $ \edgecodim \geq 2 $ and $2 \leq \dctv \leq  \frac{\extraColors}{4(1+\epsilon)}$ such that $\edgecodim + \dctv \geq 5$. 
If the line graph of $G_\tau$ is not connected, then  the statement trivially holds.  Otherwise, by \cref{rec_closefrom}, for any color $c$
and  $e, f \sim_{\tau, c} v$ we can write 
\begin{align}
      \linear_\tau (ec, fc)  &\leq \frac{\dctv-2}{\edgecodim-2} 
      \max_{gc'\in X_\tau(0): g\sim v,g\neq e,f}  A_{\tau\cup gc'} ( ec, fc) + \frac{\edgecodim-\dctv}{\edgecodim-2} 
      \max_{gc'\in X_\tau(0): g\not\sim v}  A_{\tau\cup gc'} ( ec, fc)\nonumber
     \\
     &\leq 
     \frac{\dctv-2}{\edgecodim-2} 4(1+\epsilon)\frac{\dctv-3}{\extraColors^2} +
     \frac{\edgecodim-\dctv}{\edgecodim-2} 4(1+\epsilon)\frac{\dctv-2}{\extraColors^2}     \nonumber \\
     &= \frac{4(1+\epsilon)(\dctv-2)(k-3)}{\beta^2(k-2)} \leq \frac{1}{\beta}\label{eq:Ltauefcupper}
\end{align}
where the second to last inequality follows by IH and the last inequality follows by $\dctv\leq \frac{\beta}{4(1+\eps)}$. Similarly,
\begin{align}\label{eq:edgecolLtaulower}
  \linear_\tau (ec, fc) \geq  -\frac{\dctv-2}{\edgecodim-2} 
  \min_{\substack{gc'\in X_\tau(0)\\g\sim v, g\neq e,f }} A_{\tau\cup gc'}(ec,fc)
  -\frac{\edgecodim-\dctv}{\edgecodim-2} \min_{\substack{gc'\in X_\tau(0)\\ g\not\sim v, g\neq e,f}} A_{\tau\cup gc'}(ec,fc) \geq - \frac{1}{\extraColors}.
\end{align}
Therefore, by \eqref{eq:defedgecoloringStau}
\begin{align}\label{eq:efctaucase1upper}
    \sqr_\tau (ec, fc) &=  4(1+\epsilon) \sum_{\substack{g\sim_{\tau, c} v\\ g \neq e, f }} \left[\linear^+_\tau (ec, gc) \linear^+_\tau (gc, fc)  + \linear^-_\tau (ec, gc) \linear^-_\tau (gc, fc) \right] {\leq} 4(1+\epsilon) (\dctv-2) \frac{1}{\extraColors^2}.
\end{align}
where the last inequality follows by \cref{eq:Ltauefcupper,eq:edgecolLtaulower} and that $v$ has at most $\dctv-2$ edges that can be colored by $c$, other than $e,f$. 
So combining with \eqref{eq:edgecolLtaulower}, we get $A_\tau(ec,fc) = \linear_\tau(ec,fc) + \frac{\sqr_\tau(ec,fc)}{\edgecodim-2}\geq 
-\frac{1}{\extraColors}$. Similarly, \eqref{eq:Ltauefcupper} and \eqref{eq:efctaucase1upper} gives
 \begin{align*}
 A _\tau (ec, fc) \leq \frac{4(1+\epsilon)(\dctv-2)(k-3)}{\beta^2(k-2)}
 +\frac{4(1+\epsilon) (\dctv-2)}{k-2} \frac{1}{\extraColors^2} =4(1+\epsilon)\frac{\dctv-2}{\extraColors^2}.
\end{align*}

{\bf Case \ref{secondcase}.}  
 For $\tau$ of co-dimension $k\geq 2$ we prove the claim by induction on $\dctv + \edgecodim$. The base case is when $\dctv = \edgecodim  =\frac{\extraColors}{4 (1+\epsilon)}$,  which we already proved in case \ref{firstcase} (note that we always have $k\geq \Delta_\tau(v)$). 
Now, we prove the claim for $\dctv> \frac{\extraColors}{4 (1+\epsilon) }$ (and $k\geq \dctv$). 
If the line graph of $G_\tau$ is disconnected then  the statement trivially holds.  
Otherwise, for all colors $c$ and $e, f \sim_{\tau, c} v$, we can write 
\begin{align}
     |\linear_\tau (ec, fc)|  &= \frac{\dctv-2}{\edgecodim-2} 
      \max_{gc'\in X_\tau(0): g\sim v, g\neq e,f}  |A_{\tau\cup gc'} ( ec, fc)| + \frac{\edgecodim-\dctv}{\edgecodim-2} 
      \max_{gc'\in X_\tau(0): g\not\sim v}  |A_{\tau\cup gc'} ( ec, fc)|\nonumber \\
     &\leq 
    \frac{\dctv-2}{\edgecodim-2}  \frac{1}{1.5 \extraColors - 2(1+2\epsilon)(\dctv-1) } +   \frac{\edgecodim-\dctv}{\edgecodim-2}  \frac{1}{1.5 \extraColors - 2(1+2\epsilon)\dctv }\label{eq:Lecfccase2edgecolor=}\\
    &\leq \frac{1}{1.5\beta-2(1+2\epsilon)\dctv}.\label{eq:Lecfccase2edgecolorupper}
\end{align}
where the second to last inequality follows by the IH. Furthermore, by \eqref{eq:defedgecoloringStau},
\begin{align*}
    |\sqr_\tau (ec, fc)| &=  2(1+\epsilon) \left|\sum_{g\sim_{\tau, c} v, g \neq e, f} \linear_\tau (ec, gc) \linear_\tau (gc, fc)\right|  \leq 
  2 (1+\epsilon) \frac{\dctv-2}{(1.5 \extraColors - 2(1+2\epsilon)\dctv )^2}. 
\end{align*}
where the inequality follows by \eqref{eq:Lecfccase2edgecolorupper} and that $v$ has at most $\dctv-2$ edges other than $e,f$ that can be colored by $c$.
Recall  $A_\tau = \linear_\tau + \frac{S_\tau}{\edgecodim-2}$. So, the above inequality with \eqref{eq:Lecfccase2edgecolor=}   gives 
\begin{align*}
|A_\tau (ec, fc)| &\leq \frac{\dctv-2}{\edgecodim-2}  \frac{1}{1.5 \extraColors - 2(1+2\epsilon)(\dctv-1) } +   \frac{\edgecodim-\dctv}{\edgecodim-2}  \frac{1}{1.5 \extraColors - 2(1+2\epsilon)\dctv } \\
&\quad+ \frac{2 (1+\epsilon)}{k-2}\cdot  \frac{\dctv-2}{(1.5 \extraColors - 2(1+2\epsilon)\dctv )^2}\\
&\leq \frac{1}{1.5 \extraColors - 2(1+2\epsilon) \dctv}. 
\end{align*}
where in the second inequality we used that  $\eps \leq \frac{1}{10} $ and   $\Delta \geq 2\eps^{-2}$, and that

$$1.5\beta - 2(1+2\eps) \dctv\geq 1.5\left(\frac43+4\eps\right)\Delta - 2(1+2\eps)\Delta = 6\eps \Delta - 4\eps\Delta \geq 0. $$

\end{proof}

\begin{corollary} \label{cor:A_bound}
Given a $\beta$-extra-color edge-list-coloring instance $(G,L)$
where $\beta = (\frac{4 }{3} + 4\epsilon) \Delta$ for  an $0 <\eps\leq \frac{1}{10}$ such that $ 2\eps^{-2} \leq \Delta $,
\begin{enumerate}[label=(\roman*)] 
\item For any $\tau \in X$ with $\codimfunc(\tau)\geq 2$, $v\in G$, and $e,f\sim_{\tau,c} v$, 
$|A_{\tau, \eps} (ec, fc)| \leq \frac{1}{2\epsilon\maxDeg} \leq \frac{\eps}{4}$. 
\item  For any $\tau \in X$ with $\codimfunc(\tau)\geq 3$, $v\in G$,  and $e,f\sim_{\tau,c} v$, $\sqr_{\tau, \eps} (ec, fc)\leq  \frac{(1+\epsilon)(\Delta_\tau(v)-2)}{2\epsilon^2 \maxDeg^2}$. 
\item For any $\tau\in X$ with $\codimfunc(\tau)\geq 3$, any  color $c$,  and $(e=\{u,v\},c)\in X_\tau(0)$  $c$, $S_{\tau,\eps}(ec,ec) \leq \frac{(1+\eps)(\Delta_\tau(v)+\Delta_\tau(u)-2)}{2\eps^2\Delta^2}$.
\end{enumerate}
\end{corollary}
\begin{proof}
 First, we verify (i). Using  \cref{prop:A_bound}, when $\degree_\tau (v) \leq \frac{\extraColors}{4 (1+\epsilon)}$ we have 
 $|A_{\tau}(ec,fc)|  \leq \frac{1}{\beta}$, and when $\degree_\tau (v) > \frac{\extraColors}{4 (1+\epsilon)}$ we have

$$ |A_{\tau}(ec,fc)| \leq \frac{1}{1.5(4/3+4\eps)\Delta-2(1+2\eps)\Delta} \leq  \frac{1}{2\eps \Delta},$$
where we used  $\Delta_\tau(v)\leq \Delta$.
So, by   \cref{rec_closefrom},  we get $|\linear_{\tau, \epsilon}(ec, fc)| \leq \frac{1}{2\epsilon\maxDeg}$ for any $\tau$ of co-dimension at least 3. Now, we verify (ii). If $\Delta_\tau(v)\leq \frac{\beta}{4(1+\eps)}$ then by \cref{eq:defedgecoloringStau},
\begin{align*} S_\tau(ec,fc) &= 4(1+\eps) \sum_{g\sim_{\tau,c} v, g\neq e,f} \linear^{+,v}_{\tau}(ec,gc)\linear^{+,v}_{\tau}(gc,fc) + \linear^{-,v}_{\tau}(ec,gc)\linear^{-,v}_{\tau}(gc,fc)\\
&\leq 4(1+\eps) \sum_{g\sim_{\tau,c} v, g\neq e,f}\max_{hc'\in X_\tau(0), h\neq e,g} |A_{\tau\cup hc'}(ec,gc)|\max_{hc'\in X_\tau(0), h\neq f,g} |A_{\tau\cup hc'}(gc,fc)| \tag{by \cref{rec_closefrom}}\\
&\leq 4(1+\eps) (\Delta^c_\tau(v)-2) \frac{1}{\beta^2}. 
\end{align*}
Otherwise, if $\Delta_\tau(v)\geq \frac{\beta}{4(1+\eps)}$, with a similar use of \cref{rec_closefrom},
\begin{align*}
    S_\tau(ec,fc) = 2(1+\eps) \sum_{g\sim_{\tau,c} v, g\neq e,f} \linear^{v}_{\tau}(ec,gc)\linear^{v}_{\tau}(gc,fc) 
    \leq 2(1+\eps) (\Delta^c_\tau(v)-2) \left(\frac{1}{2\eps\Delta}\right)^2, 
\end{align*}
where the first inequality uses part (i).
Finally, (ii) follows from $\frac{4(1+\eps)}{\beta^2}\leq \frac{1+\eps}{2\eps^2\Delta^2}$.
It remains to prove (iii). For a vertex $u$ let $\alpha(u)=4(1+\eps)$ if $\Delta_\tau(u)\leq \frac{\beta}{4(1+\eps)}$ and $\alpha_u=2(1+\eps)$ otherwise. By an argument similar to (ii)
\begin{align*}
    S_\tau(ec,ec) &\leq \alpha(u)\sum_{f\sim_{\tau,c}u, f\neq e} \max_{gc'\in X_\tau(0),g\neq e,f} |A_{\tau\cup gc'}(ec,fc)|^2 +\alpha(v) \sum_{f\sim_{\tau,c}v, f\neq e} \max_{gc'\in X_\tau(0),g\neq e,f} |A_{\tau\cup gc'}(ec,fc)|^2\\ 
    &\leq (\Delta_\tau(u)+\Delta_\tau(v)-2)\max\left\{\frac{4(1+\eps)}{\beta^2},\frac{2(1+\eps)}{4\eps^2\Delta^2}\right\} 
    \\
    &\leq \frac{(\Delta_\tau(u)+\Delta_\tau(v)-2)(1+\eps)}{2\eps^2\Delta^2}
\end{align*}
This completes the proof.

\end{proof}

The following lemma is a crucial part of our proof as it will help us  bound the term $M_\tau \Pi_\tau^{-1}M_\tau $ in \cref{eq:desired_ineq} effectively.

\begin{lemma}\label{decompose_A2}
 Consider a graph $G=(V, E)$, and some weight function $w: E \rightarrow \R_{\geq0}$. Let $A$ be the weighted adjacency matrix of its line graph. Then
 \begin{align*}
     A^2 \preceq 2  \sum_{v \in V} (A^v)^2, 
 \end{align*}
 where $A^v (e, f) = A(e, f)$ if $e, f \sim v$ and $0$ otherwise. 
\end{lemma}
\begin{proof}
It is enough to show that for all $x \in \mathbb{R}^{E}$, $x^\top A^2 x \leq  2   \sum_{v \in V}  x^\top (A^v)^2x.$
We have
\begin{align*}
       x^\top A^2 x = \|Ax\|_2^2 = \sum_{e  \in E} (Ax (e))^2 =  \sum_{e \in E}  \langle A_e, x \rangle^2
\end{align*}
where $A_e$ is the row indexed by $e$. Now, let $e = \{u, v\} \in E$. We can write $ \langle A_e, x \rangle = \langle (A^u)_e, x \rangle + \langle (A^v)_e, x \rangle$. Therefore, by an application of \cref{crosstermfact}  
\begin{align*}
\sum_{e \in E}  \langle A_e,x \rangle^2  \preceq  2 \sum_{e =\{u, v\} \in E}  \langle (A^u)_e, x  \rangle^2  + \langle (A^v)_e, x  \rangle^2 = 2 \sum_v \sum_{e \sim v } (A^v x (e))^2  =  2    \sum_{v \in V} x^\top(A^v)^2x.
\end{align*}
\end{proof}

Now, we apply \cref{thm:fulltechnical} to  derive  sufficient conditions on the family $\{F_\tau\}_{\tau \in X, \text{codim}(\tau)\geq 2}$  to get $\lambda_2(P_\tau) \leq \frac{\rho(F_\tau + A_\tau)}{k-1}$ for all $\tau$ of co-dimension $2 \leq k \leq m$. 

\begin{proposition}\label{prop:Fconditions-edge}

 Let $(G, L)$ be a $(\frac{4 }{3} + 4\eps) \Delta$-extra-color  edge-list-coloring instance    such that  $0 \leq \eps \leq \frac{1}{10}$ and  $\Delta \geq 2 \eps^{-2}$, and let  $(X, \pi_{m-1})$ be its associated  weighted  simplicial complex.    Suppose  that  $\{F_\tau \in \mathbb{R}^{X(0) \times X(0)}\}_{\tau \in X: \codimfunc(X) \geq 2}$ is a family of diagonal matrices supported on $X_\tau (0) \times X_\tau (0)$ such that  $F_\tau=  f_{\times} (X_\tau, \{F_{ \tau \cup \sigma} \}_{ \emptyset \subsetneq \sigma \in X_\tau(\leq \codimfunc(\tau) -3)})$ if the line graph of $G_\tau$ is connected  and otherwise, 
 

 \begin{enumerate}
     \item  For all $\tau$ of co-dimension $2$:  $F_\tau$ is defined as $F_\tau (ec, ec) = \frac{1}{(\frac{4 }{3} + 4\eps)^2\Delta^2} = \frac{1}{\beta^{2}}$ for  $ec \in X_\tau (0)$ and  $0$ on all other entries.
     \item   For all  $\tau$ of co-dimension $k \geq 3$:  $F_\tau \preceq  (\frac{(k-1)^2}{3k-1} - \frac{1}{2\epsilon\maxDeg})I^{X_\tau(0)}$, and for  any $ec\in X_\tau (0)$
    \begin{align}\label{eq:F_assump_edge}
    \sum_{gc' \in X_{\tau \cup ec}(0) } p(gc' | \tau\cup ec ) \dbound_{\tau \cup gc'} (ec, ec) \leq (\edgecodim -2) \dbound_{\tau} (ec, ec) -  \left( \frac{2+\eps}{\epsilon}\right)  \dbound^2_{\tau}(ec, ec)  - \gamma_\tau(ec), 
    \end{align}
    where  $\gamma_\tau (ec) =  \frac{(1+\eps)\Delta_\tau(e)}{2\eps^2\Delta^2}+\frac{ (1+\eps)^2(2+3\eps+\eps^2)}{\eps^5\Delta^2}$. 
  \end{enumerate}

 Then for all $\edgecodim \geq 2$ and  $\tau$ of co-dimension $k$, $\lambda_2(P_\tau ) \leq \frac{ \rho(\dbound_\tau + A_\tau )}{\edgecodim-1}$, 
where $A_\tau$ is  defined in \cref{defA}.
\end{proposition}

\begin{proof}
We prove that  the conditions of   \cref{thm:fulltechnical} hold  for  $\{\boundMatrix_{\tau}\}_{\tau \in X( \leq m-3)}$ defined as follows: 
\begin{align*}
    M_\tau \coloneqq  \Pi_\tau\frac{\dbound_\tau}{\edgecodim-1}  + \sqrt{\Pi_\tau} \frac{A_\tau}{{\edgecodim-1}} \sqrt{\Pi_\tau} \quad \forall \tau\in X, k=\codimfunc(\tau)\geq 2
\end{align*}
Note that  the condition of the theorem  holds for  any $\tau$ of co-dimension $2$ by definition. 
So, we prove the statement for $\tau$ of co-dimension at least 3. 
Assume the line graph of $G_\tau$ is disconnected. Using the definition of  $A_\tau$ and our assumption about $F_\tau$, the proof of this case  is   similar  to what we argued in  \cref{prop:Fcondition-2delta}.  
Now, assume that the line graph of $G_\tau$ is connected. 
Note   that by \cref{cor:A_bound},  the absolute value of every off-diagonal entry of $A_\tau$ is at most $\frac{1}{2\eps\Delta}$ and that there are at most $(k-1)$ non-zero entries per row. Therefore, $\sqrt{\Pi_\tau}\frac{A_\tau}{\edgecodim - 1} \sqrt{\Pi_\tau}\preceq \frac{1}{2 \epsilon \maxDeg} \Pi_\tau$. Combined with  the bound on entries of diagonal matrix $F_\tau$, 
this implies that $M_\tau \preceq \frac{\edgecodim-1 }{3\edgecodim-1} \Pi_\tau$.
 Therefore, it only remains to show that   $  \mathbb{E}_{gc \sim \pi_\tau }  M_{\tau \cup gc}   \preceq   M_{\tau} - \frac{\edgecodim -1 }{\edgecodim -2}  M_\tau\Pi_\tau^{-1} M_\tau$. This is 
equivalent to showing that
\begin{align}
    \sqrt{\Pi_{\tau}^{-1}} \mathbb{E}_{gc \sim \pi_\tau } \left[ \Pi_{\tau \cup gc}  \frac{\dbound_{\tau \cup gc}}{\edgecodim - 2 } + \sqrt{\Pi_{\tau \cup gc}} \frac{A_{\tau \cup gc}}{\edgecodim-2} \sqrt{\Pi_{\tau \cup gc}}\right]  \sqrt{\Pi_{\tau}^{-1}}  
    \preceq   \frac{\dbound_{\tau}}{\edgecodim - 1} + \frac{A_{\tau}}{\edgecodim-1} 
    -  \frac{(\dbound_{\tau} + A_{\tau})^2 } {(\edgecodim - 2)(\edgecodim - 1)}. \label{eq:targetineqindedgecoloring}
\end{align}
We proceed by first proving a lowerbound on the RHS. By two applications of \cref{crosstermfact}, we can write
\begin{align}
 (\dbound_{\tau} + A_{\tau})^2 &\preceq \left(1 + \frac{2}{\epsilon}\right) \dbound^2_{\tau} + \left( 1 + \frac{\epsilon}{2}\right) A_{\tau}^2 \nonumber\\
 &=\left(1+\frac2{\eps}\right)F_\tau^2 + \left(1+\frac\eps2\right)\left(\linear_{\tau}+\frac{\offdiag(\sqr_\tau)}{k-2}\right)^2\nonumber \\
 &\preceq 
 \left(1 + \frac{2}{\epsilon}\right) \dbound^2_{\tau} + ( 1 + \epsilon) \linear_{\tau}^2 +  \frac{(3+\eps+2/\eps)\offdiag(\sqr_\tau)^2}{(\edgecodim-2)^2}. \label{eq:FtauAtauupper}
\end{align}
We proceed by finding a diagonal matrix to upperbound  $\linear_{\tau}^2$ .  For any $c \in [q]$,    $\linear_{\tau}^c$  is the weighted adjacency matrix of a line graph. Therefore, by \cref{decompose_A2},  $(\linear_{\tau}^c)^2 \preceq 2 \sum_{v \in V} (\linear_{\tau}^{c, v} )^2$. 
Since  $\linear_{\tau}^2 = \sum_{c \in [q]}  (\linear_{\tau}^c)^2$,
we get that  
$$\linear_{\tau}^2 \preceq  2 \sum_{v \in V} (\linear_{\tau}^v)^2 \preceq 4 \sum_{v \in V}  ((\linear_{\tau}^{+, v} )^2+ (\linear_{\tau}^{-, v})^2 ).$$ 
where in the second inequality we used  \cref{crosstermfact}.
Therefore,  by definition of $S_\tau$ (see \cref{eq:defedgecoloringStau}),
\begin{align*}
( 1 + \epsilon) \linear_{\tau}^2 \preceq  \sqr_\tau = (k-2)(A_\tau - \linear_\tau) + \diag(S_\tau). 
\end{align*}
So, by \eqref{eq:FtauAtauupper}, we can lowerbound the RHS of \eqref{eq:targetineqindedgecoloring} as follows
\begin{align*}
 \frac{\dbound_{\tau}}{\edgecodim - 1} + \frac{A_{\tau}}{\edgecodim-1} 
    -  \frac{(\dbound_{\tau} + A_{\tau})^2 } {(\edgecodim - 2)(\edgecodim - 1)} 
    &\succeq  \frac{\dbound_{\tau}}{\edgecodim - 1} + \frac{\linear_{\tau}}{\edgecodim-1} -   \frac{(1 + \frac{2}{\epsilon})\dbound^2_{\tau} }{(\edgecodim -1)(\edgecodim -2) }
   \\& - \frac{ \diag (S_\tau) }{(\edgecodim -1)(\edgecodim -2) }  - 
    \frac{(3+\eps+\frac2\eps)\offdiag(S_\tau)^2}{(k-1)(k-2)^3}.
\end{align*}
On the other hand, by definition of $\linear_\tau$ (see \eqref{eq:Ltaudefedgecolor}),  the LHS of \eqref{eq:targetineqindedgecoloring} is equal to
\begin{align*}
   \mathbb{E}_{gc \sim \pi_\tau } \left[ \Pi_{\tau}^{-1}\Pi_{\tau \cup gc}  \frac{\dbound_{\tau \cup gc}}{\edgecodim - 2 }\right] + \frac{\linear_{\tau}}{\edgecodim-1},
\end{align*}
and   
\begin{align*}
    \mathbb{E}_{gc \sim \pi_\tau } \left[ \Pi_{\tau}^{-1}\Pi_{\tau \cup gc}  \frac{\dbound_{\tau \cup gc}}{\edgecodim - 2 }\right] (ec, ec) =    \sum_{gc' \in X_{\tau \cup ec} (0) } p(gc' | \tau\cup ec ) \dbound_{\tau \cup gc'} (ec, ec).
\end{align*}
Comparing this with the assumption (see \cref{eq:F_assump_edge}), and letting $\gamma_\tau(ec) = 0$ for all $ec \notin X_\tau (0)$,  it is enough to show that 
\begin{align*}
\diag(\gamma_\tau)\succeq  \diag(S_\tau)+\frac{(3+\eps+\frac{2}{\eps})\offdiag(S_\tau)^2}{(k-2)^2}.
\end{align*}
First, notice,
$$ \frac{\offdiag(S_\tau)^2}{(k-2)^2} \preceq 
\frac{\norm{ \offdiag(S_\tau)}_\infty^2 I^{X_\tau(0)}}{(k-2)^2} 
\preceq \frac{(1+\eps)^2(\Delta-2)^24(\Delta-1)^2}{4\eps^4\Delta^4(k-2)^2} I^{X_{\tau}(0)} \preceq \frac{(1+\eps)^2}{\eps^4\Delta^2}I^{X_{\tau}(0)},$$
where the second inequality is by \cref{fact:rowsum}, noting that  by part (ii) of \cref{cor:A_bound}, every off-diagonal entry of $S_\tau$ is at most $\frac{(1+\eps)(\Delta-2)}{2\eps^2\Delta^2}$ and that there are at most $2(\Delta-1)$ non-zero entries per row. 
Finally, the statement follows from part (iii) of \cref{cor:A_bound} which shows  $S_\tau(ec,ec)\leq \frac{(1+\eps)\Delta_\tau(e)}{2\eps^2\Delta^2}$ for any $ec\in X_\tau(0)$.
\end{proof}

With this in hand, we prove \cref{thm:edge-coloring}.



\begin{proof}[Proof of \cref{thm:edge-coloring}]
For any $ec = \{u, v\} c \in  X_\tau(0)$ define,
\begin{align*}
    F_\tau (ec, ec) \coloneqq
    \begin{cases*}
       0 & if $\Delta_\tau(e) = 0$ \\
        f_1 (\degree_\tau(g))  & if $\Delta_\tau(e) = 1, g \sim e$,\\
       f_2 (\degree_\tau(e))  & if $\Delta_\tau(e) \geq  2$. 
    \end{cases*}
\end{align*}
where  $f_1(i) \coloneqq \frac{1}{(\frac{4}{3} + 4\eps)^2 \Delta^2} + \frac{  ( 4 \epsilon^{-5} + 0.6 \epsilon^{-2}) \sum_{k=1}^{i-1} \frac{1}{k}}{\maxDeg^2 }$ for any $i \geq 2$,  and $f_2(i) \coloneqq \frac{5 \epsilon^{-5} \ln (\maxDeg)  + (4\epsilon^{-5}+ \epsilon^{-2} i )\sum_{k=1}^{i-1} \frac{1}{k} }{\maxDeg^2 }$ for  $i \geq 2$. We prove that this satisfies the  conditions of \cref{prop:Fconditions-edge}.
Then, the statement follows from the fact that
$$ \lambda_2(P_\tau)\leq \frac{\rho(F_\tau+A_\tau)}{k-1} \leq \frac{\eps + \frac{1}{\eps}}{k-1},$$
where the last inequality follows by \eqref{eq:f2edgecoloringbound} below and  the fact that every entry of $A_\tau$ is at most $\frac{1}{2\eps\Delta}$ (by \cref{cor:A_bound}) and that every row of $A_\tau$ has at most $2(\Delta-1)$ non-zero entries.

The condition for links of co-dimension $2$ holds by definition as $f_1(1) = \frac{1}{(4/3+4\eps)^2\Delta^2}$. Assume that  $k \geq 3$ and  $\tau$ is of co-dimension $k$. 
Similar to the proof of \cref{findF-2delta}, when the line graph of $G_\tau$ is disconnected, the condition holds. Now, assume that the line graph of $G_\tau$ is connected. It follows that, for all $1\leq i\leq 2\Delta$, 
\begin{align}
f_2(i)  \leq \frac{9\eps^{-5}(\ln(\Delta)+1)}{\Delta^2}+\frac{2\eps^{-2}(\ln(\Delta)+2)}{\Delta}
  \underset{\frac{\ln^2(\Delta)}{\Delta}\leq \frac{\eps^3}{15}}{\leq} \frac{3\eps}{100}+\frac{2\eps}{15} \underset{\eps\leq 0.1, k\geq 2}\leq \frac{(k-1)^{2}}{3 k-1}-\frac{1}{2\epsilon \Delta}.\label{eq:f2edgecoloringbound}
\end{align}
A similar inequality holds for $f_1$ and $1\leq i\leq 2\Delta$.
It remains to check the condition \cref{eq:F_assump_edge} in \cref{prop:Fconditions-edge}. We need to show that  for any $ec   \in X_\tau(0)$, 
\begin{align*}
    \sum_{fc'\in X_{\tau \cup ec}(0)} p(fc'|\tau \cup ec) F_{\tau\cup fc'} (ec, ec) \preceq (k-2) F_\tau (ec, ec) - \left(1+\frac{2}{\epsilon}\right)F_\tau^2(ec, ec) - \gamma_\tau (ec),
\end{align*}
for 
\begin{align*}
\gamma_\tau (ec)=\frac{(1+\eps)\Delta_\tau(e)}{2\eps^2\Delta^2}+\frac{ (1+\eps)^2(2+3\eps+\eps^2)}{\eps^5\Delta^2}
    \underset{\eps\leq 0.1}{\leq} \frac{0.6\Delta_\tau(e)\eps^{-2}}{\Delta^2}+ \frac{3\eps^{-5}}{\Delta^2}.
\end{align*}

{\bf Case 1: $\Delta_\tau(e) = 1, g \sim_\tau e$.} 
 Since the line graph of  $G_\tau$ is connected   and $\tau$ is of co-dimension at least $3$, $\Delta_\tau(g) \geq 2$. So, it is enough to show that,
\begin{align}
    \sum_{fc'\in X_\tau(0)} p(fc'|\tau\cup ec) F_{\tau\cup fc'} (ec, ec) &= (\Delta_\tau (g)-1) f_1(\Delta_\tau (g) -1) + (k- \Delta_\tau (g) - 1) f_1(\Delta_\tau(g))\nonumber \\& \leq (k-2)f_1(\Delta_\tau(g)) - \left(1+\frac{2}{\epsilon}\right) f_1^2(\Delta_\tau(g)) -  \frac{0.6\epsilon^{-2} + 3\epsilon^{-5}}{ \maxDeg^2}.\label{eq:edge_desired1}
\end{align}
Now, note that 
\begin{align*}
&
(k-2)f_1(\Delta_\tau(g)) -(\Delta_\tau (g)-1) f_1(\Delta_\tau (g) -1) - (k- \Delta_\tau (g) - 1) f_1(\Delta_\tau(g))  \\& = 
(\Delta_\tau (g) -1)(f_1(\Delta_\tau(g))  - f_1(\Delta_\tau(u) -1 )) = \frac{0.6\epsilon^{-2} + 4 \epsilon^{-5}}{ \maxDeg^2}. 
\end{align*}
Furthermore,
\begin{align*}
 \left(1+ \frac{2}{\epsilon}\right) f_1^2 (\Delta_\tau(g))  \underset{\eps\leq 0.1}{\leq} \frac{2.1}{\eps}\left(\frac{5 \epsilon^{-5} \ln \Delta}{\Delta^2}\right)^2   \underset{\frac{\ln^2(\Delta)}{\Delta}\leq \frac{\eps^3}{15}}{\leq}  \frac{ \epsilon^{-5}}{  \maxDeg^2}.
 \end{align*}
 Putting these together, we get \cref{eq:edge_desired1}.\\
{\bf Case 2:  $\Delta_\tau(e) \geq 2$.}
For convenience in writing the recursion, let $f_2(1) = \frac{5 \epsilon^{-5} \ln (\maxDeg) }{\maxDeg^2 }$.
Following similar calculations,  it is  enough to show that 
\begin{align}\label{eq:edge_desired2}
(\Delta_\tau (e)-1)f_2(\Delta_\tau (e))  - \Delta_\tau (e)f_2(\Delta_\tau (e)-1) \geq  \left(1+\frac{2}{\epsilon}\right) f_2^2(\Delta_\tau (e)) + \frac{0.6\epsilon^{-2} \Delta_\tau (e) +  3 \epsilon^{-5}}{ \maxDeg^2}.
\end{align}
Note that in the LHS of the above equation we should write $f_1(\Delta_\tau(g))$ if $\Delta_\tau(e)-1=1$ and $g$ is the only remaining neighbour (of $e$), but since $f_1(i)\leq \frac{5 \epsilon^{-5} \ln (\maxDeg) }{\maxDeg^2 }=f_2(1)$ for all $1\leq i\leq 2\Delta$ the above inequality is valid. 
Note that, by definition  
\begin{align*}
(\Delta_\tau (e)-1)f_2(\Delta_\tau (e))  - \Delta_\tau (e)f_2(\Delta_\tau (e)-1 ) =    \frac{ \epsilon^{-2} \Delta_\tau (e) +   4\epsilon^{-5}}{ \maxDeg^2}.
\end{align*}
Furthermore,  $\frac{\ln^2(\Delta)}{\Delta } \geq \frac{\eps^3}{15}$ and $\eps\leq 0.1$ imply that $\ln \Delta \geq 10$, and we can write 
\begin{align*}
    \left(1+ \frac{2}{\epsilon}\right) f_2^2 (\Delta_\tau (e))  &\underset{\eps\leq 0.1}{\leq} 
    \frac{2.1(\ln\Delta+2)^2}{\epsilon \Delta^4}\cdot \left( \eps^{-2}\Delta_\tau(e)+9\eps^{-5}\right)^2\\
    &\underset{\cref{crosstermfact}}{\leq}
\frac{2.1(\ln\Delta+2)^2}{\epsilon \Delta^4}\cdot \left(1.2 (\eps^{-2}\Delta_\tau(e))^2 + 6 (9\eps^{-5})^2\right)  
\\&\underset{\ln\Delta\geq 10, \frac{\ln^2(\Delta)}{\Delta}\leq \frac{\eps^3}{15}}{\leq}    \frac{0.4 \epsilon^{-2} \Delta_\tau (e) +   \epsilon^{-5}}{ \maxDeg^2}.
\end{align*}
 This finishes the proof of  \cref{eq:edge_desired2}.
\end{proof}

\printbibliography
 
\appendix
\section{Proofs from \cref{sec:vertexcoloring} }
\label{sec:appendixvercolor}

\begin{proof}[Proof of \cref{thm:vertex-coloring-2delta}]\label{findF-2delta}
 For each $\tau$ of co-dimension at least $2$, let $F_\tau \in \mathbb{R}^{X(0) \times X(0)}$ be a diagonal matrix supported on $X_\tau (0) \times X_\tau(0)$ defined as  
follows: for any
$vc \in X_\tau(0)$, 
\begin{align*}
    F_\tau(vc, vc) \coloneqq
    \begin{cases*}
        0 & if $\degree_\tau(v) =0$, \\
       f_1(\Delta_\tau(u)) & if $\degree_\tau(v) =1 $ and $u \sim_\tau v$, \\
       f_{2}(\Delta_\tau (v))   & if $\degree_\tau(v) \geq 2$, 
    \end{cases*}
\end{align*}
where $ f_1 (i) = \frac{1}{(1+\eps)\Delta} + \frac{1+ 2 \sum_{j=1}^{i-1} \frac{1}{j}}{(1+ \eps)^2\Delta^2} $ for $i\geq 1$ and $f_{ 2} (i)= \frac{i}{(1+\frac{\eps}{2})\Delta - (i -1) - \frac{4}{\eps} \sum_{j=1}^{i-1} \frac{1}{j} }$ for $i\geq 2$. 
We show that  the  conditions of \cref{prop:Fcondition-2delta} hold for     $\{F_\tau \}_{\tau \in X: \codimfunc(\tau) \geq 2 }$.  Then, the statement follows  from the fact that $\rho(F_\tau) \leq \frac{5}{2\eps}$. This is true because  $\frac{\ln (\Delta) + 2}{\Delta } \leq \frac{\eps^2}{40}$ implies that for any $1 \leq i \leq \Delta$ the denominator of $f_2(i) $ is at least $\frac{2\eps}{5} \Delta$ and thus  $f_1 (i) \leq  f_2(i) \leq \frac{5}{2\eps}$.
\par
The condition for links of co-dimension $2$ holds by definition. Assume $\tau$ is of co-dimension  at least $3$. 
When $G_\tau$ is disconnected one can  check the condition holds because of the fact that the degrees of  vertices of   a connected component do not change by   removing vertices from other connected components of the graph. Now, assume that $G_\tau$ is connected. 
Note that for every $v \in V_\tau$,  we have

$$\frac{(k-1)^2}{3k-1} \geq \frac{\Delta_\tau (v)}{5} \underset{ \frac{\ln (\Delta) + 2}{\Delta } \leq \frac{\epsilon^2}{40}}{\geq}  \frac{8}{\epsilon^2} \underset{f_2(\Delta_\tau (v))\leq \frac{5 }{2\eps},\eps \leq 1}{\geq} f_2(\Delta_\tau (v)) \geq f_1(\Delta_\tau (v)).$$

Therefore, it is enough to show that   $\sum_{uc'\in X_{\tau \cup vc} (0)} p(uc'|\tau\cup vc) F_{\tau\cup uc'} (vc, vc) \preceq (k-2) F_\tau (vc) - F_\tau^2(vc)$ for any $vc \in X_\tau (0)$.
\par
{\bf Case 1:  $\Delta_\tau (v) = 1$, and $ u \sim_\tau v$.}  Since $G_\tau$ is connected and $\tau$ is of co-dimension at least $3$, $\Delta_\tau(u) \geq 2$. We have
\begin{align*}
    \sum_{wc' \in X_\tau(0)} p(wc'|\tau\cup vc) F_{\tau\cup wc'} (vc, vc) =  (\Delta_\tau (u)-1) f_1(\Delta_\tau (u) -1 ) + (k-\Delta_\tau (u)-1) f_1(\Delta_\tau (u) ).  
\end{align*}
On the other hand, $\frac{\ln (\Delta) + 2}{\Delta } \leq \frac{\eps^2}{40}$ and $\eps\leq 1$ imply that $\frac{1 + 2 \sum_{j=1}^{i-1} \frac{1}{j}}{(1+\eps) \Delta } \leq  \frac{1}{20(1+\eps)}$ for any $1 \leq i \leq \Delta$. Therefore, 
\begin{align*}
 F_\tau^2(vc, vc)& =   f_1^2(\Delta_\tau(u)) \leq  \left(\frac{1}{(1+\eps)\Delta} + \frac{1}{20(1+\eps)\Delta}\right)^2  \leq \frac{2}{(1+\eps)^2\Delta^2}. 
\end{align*}
Therefore $(k-2) F_\tau (vc) - F_\tau^2(vc) \geq (k-2)f_1 (\Delta_\tau (u))  - \frac{2}{(1+\eps)^2\Delta^2}$ and thus it is enough to show that
\begin{align*}
(\Delta_\tau (u)-1) \left(f_1 \left(\Delta_\tau (u)  \right)- f_1\left(\Delta_\tau (u) -1 \right)\right) \geq  \frac{2}{(1+\eps)^2\Delta^2}. 
\end{align*}
But this inequality holds with equality.

{\bf Case 2: $\Delta_\tau (v) \geq 2$. }
One can check that  $\frac{\ln (\Delta)+1}{\Delta} \leq \frac{\eps^2}{40}$ implies 
$\frac{1}{(1+\eps)\Delta} + \frac{2 \sum_{j=1}^{\Delta-1} \frac{1}{j}+1 }{(1+\eps)^2\Delta^2}   \leq  \frac{1}{(1+\frac{\eps}{2})\Delta}$. For convenience define $f_2(1) = \frac{1}{1+\frac{\eps}{2}}$ and notice that $f_1(i)\leq f_2(1)$ for any $1 \leq i \leq \Delta$. We want to show that 
\begin{align*}
    \sum_{uc' \in X_\tau(0)} p(uc'|\tau\cup vc) F_{\tau\cup uc'} (vc, vc) &=  \Delta_\tau (v) f_{ 2} (\Delta_\tau (v) -1 ) + (k-\Delta_\tau (v)-1) f_{ 2} (\Delta_\tau (v) ) \\&\leq (k-2) f_{ 2} (\Delta_\tau (v)) - f^2_{ 2} (\Delta_\tau (v)).  
\end{align*}
Let $i \coloneqq \Delta_\tau (v)$. We have
\begin{align*}
     (i-1) f_{ 2} (i)-i f_{ 2} (i -1 )   = \frac{i(i-1)+ \frac{4}{\eps}i }{\left((1+\frac{\eps}{2})\Delta - (i -1) - \frac{4}{\eps}  \sum_{j=1}^{i-1}\frac{1}{j}  \right)\left((1+\frac{\eps}{2})\Delta - (i -2) - \frac{4}{\eps}  \sum_{j=1}^{i-2}\frac{1}{j} \right)}. 
\end{align*}
Therefore
\begin{align*}
    &(i-1) f_{ 2} (i) - i f_{ 2} (i -1 )  - f_2^2(i) = \frac{(\frac{4}{\eps}-1)i \left( (1+\frac{\eps}{2})\Delta  - (i-1) - \frac{4}{\eps}  \sum_{j=1}^{i-1}\frac{1}{j}\right) -i^2(1 + \frac{4}{\eps} \frac{1}{i-1})  }{\left((1+\frac{\eps}{2})\Delta - (i -1) - \frac{4}{\eps} \sum_{j=1}^{i-1}\frac{1}{j}  \right)^2\left((1+\frac{\eps}{2})\Delta - (i -2) - \frac{4}{\eps}  \sum_{j=1}^{i-2}\frac{1}{j} \right)}. 
\end{align*}
The denominator is  positive for $1 \leq i \leq \Delta$ and for the numerator we have
\begin{align*}
(\frac{4}{\eps}-1)i \left( (1+\frac{\eps}{2})\Delta  - (i-1) - \frac{4}{\eps}  \sum_{j=1}^{i-1}\frac{1}{j}\right) -i^2(1 + \frac{4}{\eps} \frac{1}{i-1})  &\underset{i\leq \Delta}\geq  
(1-\frac{\eps}{2})i\Delta- \frac{16 i}{\eps^2} \left(\ln (\Delta)+1 \right)  - \frac{4i^2}{\eps(i-1)} 
\end{align*}
Canceling out an $i$, and using $\frac{\ln (\Delta) + 2}{\Delta } \leq \frac{\eps^2}{40}$ and $\eps\leq 1$ the RHS is non-negative.   
\end{proof}

\begin{proof}[Proof of \cref{thm:vertex-coloring-tree}]
 As before we make the tree rooted at an arbitrary vertex. For any $k \geq 2$, any $\tau$ of co-dimension $k$, let  $F_\tau \in \mathbb{R}^{X(0) \times X(0)}$ be a diagonal matrix supported on $X_\tau(0)\times X_\tau(0)$ defined  as follows:  for any
$vc \in X_\tau(0)$,  
\begin{align*}
    F_\tau(vc, vc) \coloneqq
    \begin{cases*}
        0 & if $\degree_\tau(v) =0$, \\
        f_1( \Delta_\tau (u))   & if $\degree_\tau(v) = 1$, $u \sim v$, and $v$ is the root of a component of $G_\tau$,\\
        f_2(\Delta_\tau (v))   & if $\degree_\tau(v) \geq 2$ and $v$ is the root of a component of $G_\tau$,\\
        f_3(\Delta_\tau (v), \Delta_\tau(a(v)))  & if $\degree_\tau(v) \geq 1$ and $a(v) \in V_\tau$, 
    \end{cases*}
\end{align*}
where $a(v)$ is the immediate ancestor of $v$ in $G$, and $f_1(i) \coloneqq \frac{5  (\sum_{j=1}^{i-1} \frac{1}{j})+1 }{\epsilon^2 \Delta^2}$ for $i \geq 1$,    $f_2(i)   \coloneqq \frac{5( \ln (\Delta) +1  +   i\sum_{j=1}^{i-1} \frac{1}{j}) }{\epsilon^2 \Delta^2}$ for $i \geq 2$, and 
$f_3(i, j) \coloneqq \frac{5 (\ln (\Delta) +1 + j  +  i \sum_{k=1}^{i-1} \frac{1}{k} }{\epsilon^2 \Delta^2}$ for $i, j\geq 1$. We prove that  the  conditions of \cref{prop:Fcondition-trees} hold for     $\{F_\tau \}_{\tau \in X: \codimfunc(\tau) \geq 2}$.
Then, the statement follows from the fact that
$$ \lambda_2(P_\tau)\leq \frac{\rho(F_\tau+A_\tau)}{k-1} \leq \frac{\rho(F_\tau) + \frac{1}{\eps}}{k-1} + \frac{\frac{1}{20} + \frac{1}{\eps}}{k-1},$$
where the second to last inequality follows by \eqref{eq:ftreebound} below and  the fact that every entry of $A_\tau$ is at most $\frac{1}{\eps\Delta}$  and that every row of $A_\tau$ has at most $\Delta$ non-zero entries.
The last inequality uses that since  $\frac{\ln^2 (\Delta)}{\Delta } \leq \frac{\epsilon^2}{100}$ and $\eps \leq 1$, for all $1\leq i \leq \Delta$ 
\begin{equation}\label{eq:ftreebound}
    f_1(i), f_2(i), f_3(i) \leq \frac{5 (\ln (\Delta)+2)}{\epsilon^2 \Delta} \leq \frac{1}{20}.
\end{equation}

The condition for links of co-dimension $2$ holds by definition. 
Assume that $\tau$ is of co-dimension   $k\geq 3$. 
Similar to the proof of \cref{thm:vertex-coloring-2delta}, when $G_\tau$ is disconnected, the condition holds. Now, assume $G_\tau$ is connected.  By \cref{eq:ftreebound}, for all $vc \in X_\tau(0)$  $F_\tau (vc, vc) \leq \frac{1}{20} \leq \frac{(k-1)^{2}}{3 k-1}-\frac{1}{\beta}$. Therefore, it is enough to show that  for any $vc \in X_\tau(0)$, 
\begin{align*}
    \sum_{wc'\in X_{\tau \cup vc}(0)} p(wc'|\tau \cup vc ) F_{\tau\cup wc'} (vc, vc) \preceq (k-2) F_\tau (vc, vc) - 2F_\tau^2(vc, vc) - \gamma_\tau ( vc),
\end{align*}
for $\gamma_\tau(vc)$ defined in \cref{prop:Fcondition-trees}. 

{\bf Case 1: $\degree_\tau(v) = 1$ and $v$ is the root of a component of $G_\tau$. } Let $u$ be the only neighbor of $v$. Since $G_\tau$ is connected   and $\tau$ is of co-dimension at least $3$, $\Delta_\tau(u) \geq 2$. We need to show that 
\begin{align}
    \sum_{wc'\in X_{\tau \cup  vc} (0)} p(wc'|\tau \cup vc) F_{\tau\cup wc'} (vc, vc) &= (\Delta_\tau (u)-1) f_1(\Delta_\tau (u) -1) + (k- \Delta_\tau (u) - 1) f_1(\Delta_\tau(u))\\& \leq (k-2)f_1(\Delta_\tau(u)) - 2 f_1^2(\Delta_\tau(u)) - \frac{4}{\epsilon^2 \Delta^2}\label{eq:findF_tree_desired1}. 
\end{align}
But,
\begin{align*}
&(k-2)f_1(\Delta_\tau(u))  - (\Delta_\tau (u)-1) f_1(\Delta_\tau (u) -1) -  (k- \Delta_\tau (u) - 1)f_1(\Delta_\tau(u)) \\&= 
(\Delta_\tau (u) -1)(f_1(\Delta_\tau(u))  - f_1(\Delta_\tau(u) -1 ))=  \frac{5}{\epsilon^2\Delta^2}. 
\end{align*}
Furthermore, one can see that $\frac{\epsilon^2}{100}$ implies that  $ 2 f^2_1(\Delta_\tau(u)) \leq \frac{50(\ln(\Delta)+2)^2}{\eps^4\Delta^4}
 \leq  \frac{1}{\epsilon^2\Delta^2}$.
This completes the proof of \cref{eq:findF_tree_desired1}.

{\bf Case 2: $\degree_\tau(v) \geq  2$, and $v$ is the root of a component of $G_\tau$.}
Note that $f_1$ is bounded above by $\frac{5 (\ln (\Delta)+1)}{\epsilon^2 \Delta^2}$. 
For convenience in writing the recursion, let $f_2(1) = \frac{5( \ln (\Delta)+1)}{\epsilon^2 \Delta^2}$.
Following similar calculations, it is  enough to show that 
\begin{align}\label{eq:findF_desired2}
(\Delta_\tau (v) -1)f_2(\Delta_\tau(v))  - (\Delta_\tau (v))f_2(\Delta_\tau(v) -1 ) \geq 2 f_2^2(\Delta_\tau(v) + \frac{4 \Delta_\tau (v) }{ \epsilon^2 \Delta^2}.
\end{align}
But one can see that, by definition 
\begin{align*}
(\Delta_\tau (v) -1)f_2(\Delta_\tau(v))  - (\Delta_\tau (v))f_2(\Delta_\tau(v) -1 )  = \frac{5 \Delta_\tau(v)}{\epsilon^2\Delta^2}, 
\end{align*}
And, 

\begin{align}\label{eq:Fsqr_tree}
    2 f^2_2(\Delta_\tau(v)) \leq \frac{50 (\Delta_\tau(v)+2)^2\ln^2(\Delta)}{\epsilon^4\Delta^4}  \underset{\Delta \geq 100}{\leq} \frac{55 \Delta^2_\tau(v)\ln^2(\Delta)}{\epsilon^4\Delta^4} 
\underset{\frac{\ln^2( \Delta)}{\Delta} \leq \frac{\epsilon^2}{100}}{\leq}  \frac{0.55 \Delta^2_\tau(v)}{\epsilon^2\Delta^3} \leq  \frac{ \Delta_\tau(v)}{ \epsilon^2\Delta^2}.
\end{align} 

This finishes the proof of \cref{eq:findF_desired2}.

{\bf Case 3: $\Delta_\tau (v) \geq 1$ and $a(v) \in V_\tau$.}
For convenience in writing the recursion, for $1 \leq i \leq \Delta$, let $f_3(i, 0) = \frac{5( \ln (\Delta)+1 +  i \sum_{k=1}^{i-1} \frac{1}{k})}{\epsilon^2 \Delta^2}$ and note that $\max_{1 \leq j \leq \Delta} f_1 (j) \leq f_2 (i) \leq f_3(i, 0)$.  
Similar to the previous cases, after simplifying the recursion, one can see that it is enough that  for $i = \Delta_\tau (v) $ and $j = \Delta_\tau (a(v))$ 
\begin{align*}
   (i - 1)f_3(i , j) -  i f_3(i -1, j) +
    (j-1) (f_3(i, j) -f_3(i, j-1)) \geq    2 f_3^2 (i,j )  + \frac{4 i }{ \epsilon^2 \Delta^2} +  \frac{4 (j -1)}{ \epsilon^2 \Delta^2} . 
\end{align*}
Now, note that 
\begin{align*}
  (i - 1)f_3(i , j) -  i f_3(i -1, j) +
    (j-1) (f_3(i, j) -f_3(i, j-1)) =   \frac{5 i }{ \epsilon^2 \Delta^2} + \frac{5 (j -1)}{ \epsilon^2 \Delta^2}. 
\end{align*}
Furthermore,  
$$ 2 f^2_3(i)  = 2 (f^2_2(i) +   \frac{5(j-1)}{\epsilon^2 \Delta^2})^2  \leq 2.5 f_2^2(i) + \frac{250(j-1)^2}{\epsilon^4 \Delta^4}  \underset{\Delta \geq 100 \eps^{-2}}{\leq }   \frac{i}{\epsilon^2\Delta^2} +    \frac{j-1}{\epsilon^2 \Delta^2}, $$
where the last inequality  uses the calculations in \cref{eq:Fsqr_tree}. 
This finishes the proof.
\end{proof}

\end{document}